\documentclass[10pt, a4paper]{article}
\usepackage[inner=2.25cm, outer=2.25cm, top=2.25cm, bottom=2cm]{geometry}

\usepackage{amsmath, amsfonts, amssymb}
\usepackage{amsthm}
\usepackage{mathtools}
\usepackage{graphicx}

\usepackage{mathtools} 
\usepackage{amssymb, bm, fixmath}
\usepackage{float}
\usepackage{subfigure}
\usepackage{tikz}
\usetikzlibrary{quantikz}
\usepackage{algorithm,algpseudocode}
\usepackage{booktabs}
\usepackage{cite}

\newcommand{\R}{\mathbb R}
\newcommand{\Z}{\mathbb Z}

\newcommand{\spn}{\text{span}_{\R}}
\newcommand{\im}{\mathrm{Im}}
\renewcommand{\Re}{\mathsf{I\!Re}}
\renewcommand{\Im}{\mathsf{I\!Im}}
\newcommand{\partialk}{\frac{\partial}{\partial\theta_k}}
\newcommand{\orthog}{^{\perp_{\Re\langle\cdot,\cdot\rangle}}}

\newcommand{\U}{\text{U}}
\newcommand{\SU}{\text{SU}}
\newcommand{\PSU}{\text{PSU}}
\newcommand{\PU}{\text{PU}}
\newcommand{\CU}{\text{CU}}
\newcommand{\GL}{\text{GL}}
\newcommand{\LT}{\text{LT}}
\newcommand{\un}{\mathfrak{u}}
\newcommand{\su}{\mathfrak{su}}

\newcommand{\tr}{\mathsf{Tr}}

\newcommand{\dagg}{^{\dagger}}
\newcommand{\Vct}{V_{\mathsf{ct}}}
\newcommand{\ct}{{\mathsf{ct}}}
\newcommand{\btheta}{{\mathbold{\theta}}}
\newcommand{\cart}{\mathsf{cart}}
\newcommand{\sequ}{\mathsf{sequ}}
\newcommand{\spin}{\mathsf{spin}}
\newcommand{\cng}{\gategroup[2,steps=1,style={dashed,rounded corners,fill=white, inner xsep=2pt},background]{}}
\newcommand{\cnq}{\gategroup[3,steps=1,style={dashed,rounded corners,fill=white, inner xsep=2pt},background]{}}

\newtheorem{lemma}{Lemma}
\newtheorem{theorem}{Theorem}
\newtheorem{remark}{Remark}
\newtheorem{conjecture}{Conjecture}
\newtheorem{definition}{Definition}

\newcommand{\cmt}[1]{{\color{black}{ #1 }}}

\makeatletter
\def\@fnsymbol#1{\ensuremath{\ifcase#1\or *\or **\or
   \mathsection\or \mathparagraph\or \|\or **\or \dagger\dagger
   \or \ddagger\ddagger \else\@ctrerr\fi}}
    \makeatother

\newcommand{\rev}{}

\begin{document}

\title{\sf Best Approximate Quantum Compiling Problems}

\author{Liam Madden\thanks{Email: Liam.Madden@colorado.edu, University of Colorado Boulder}\qquad Andrea Simonetto\thanks{Email: {andrea.simonetto@ensta-paris.fr}, {UMA, ENSTA Paris, Institut Polytechnique de Paris, 91120 Palaiseau, France. The authors performed part of the work while at IBM Quantum, IBM Research Europe.}}}

\maketitle
\begin{abstract}
  We study the problem of finding the best approximate circuit that is the closest (in some pertinent metric) to a target circuit, and which satisfies a number of hardware constraints, like gate alphabet and connectivity. We look at the problem in the CNOT$+$rotation gate set from a mathematical programming standpoint, offering contributions both in terms of understanding the mathematics of the problem and its efficient solution. Among the results that we present, we are able to derive a 14-CNOT 4-qubit Toffoli decomposition from scratch, and show that the Quantum Shannon Decomposition can be compressed by a factor of two without practical loss of fidelity. 
\end{abstract}




\maketitle

\section{Introduction}

With the steady advances in quantum hardware and volume~\cite{jurcevic2020demonstration}, quantum computing is well on track to become widely adopted in science and technology in the near future. One of the core challenges to enable its use is the availability of a flexible and reliable quantum compiler, which can translate any target quantum circuit into a circuit that can be implemented on real hardware with gate set, connectivity, and length limitations. 

Since the celebrated Solovay-Kitaev theorem~\cite{nielsenbook,dawson2005}, quantum compiling has been a rich research area. Works have investigated how to efficiently map different gates into canonical (universal) gate sets up to an arbitrary accuracy\cmt{~\cite{kliuchnikov2012fast,ross2014,sarnakletter,selinger2013nqubit,redu2008,redu2013,redu2018,vala2004}}, or how to ``place'' the target circuit onto the real connectivity-limited hardware~\cmt{\cite{place2008,place2017,place2018,place2018b,place2019,place2019b,place2019c,place2019d,tket,place2020,place2020b,murali2019full,tan2020optimal,rogers2021synthesis,Giacomo2021}}.

A more holistic strategy in quantum compiling has been the construction of universal parametric circuits to serve as templates to compile any target circuit. This line of research, which we will call decomposition-based, focuses primarily on templates based on the versatile CNOT$+$rotation gate set~\cmt{\cite{barenco1995,bullock2003arbitrary,shende2004lowerbound,shende2006,vatan2004optimal,drury2008constructive,nakajima2005new}}. These works unveiled fundamental lower bounds on the number of CNOTs that almost all target circuits require in order to be compiled in such a gate set and delivered a constructive method for doing so, the quantum Shannon decomposition (QSD), which was only a factor of two off from the lower bound of efficiency. This research is formalized mainly in the language of Lie theory as a recursive sequence of Cartan decompositions. Despite the constructive and sound theory, the QSD decomposition does not have the flexibility of trading off precision and length.

The quantum compiling problem in its essence can be cast as an optimization problem in the space of unitary matrices. Here one needs to find a unitary matrix that can be realized in hardware (with various constraints, e.g., gate set and connectivity) that is the ``closest'' to a target unitary (i.e., the circuit that one wants to realize, or compile). Here ``closest'' is intended with respect to a pertinent metric. On the one hand, this optimization problem could encompass the whole quantum compiling research; on the other hand, it is a very difficult mathematical problem and even for a small number of qubits cannot be stored in memory. Recently, a series of papers~\cite{cincio2018,khatri2019}, have revised this optimization-based compiling approach with some simplifying assumptions and heuristics, \cmt{and have introduced the idea of computing the cost and its gradients in a quantum-assisted way}. \cmt{With the same optimization lens, but with other classical heuristics, the works~\cite{younis2020qfast,younis2021qfast} have looked at hierarchical compilations, whereby one need not compile a unitary directly to a two-qubit gate circuit, but can instead start with higher-qubit gates, and then down recursively to the two-qubit gate target.   } \rev{Finally, the recent works \cite{squander,rakyta2021approaching} arrange the CNOTs to explicitly decouple the qubits one at a time using a particular cost function. Moreover, they provide numerical evidence that their approach can compile arbitrary target circuits very close to the lower bound on the number of CNOTs.}

In this paper, we aim at analyzing the optimization approach in more depth and offering some sound evidence on the justification of critical assumptions and solution methods. Further, we aim at helping to bridge the gap between Cartan decomposition-based research and optimization-based compiling. Our approach consists in formulating best approximate compiling problems, which trade-off exact compilation with constraint violation. In particular, we offer the following contributions.

\begin{itemize}
    \item We start by analyzing the quantum compiling mathematical problem in the CNOT$+$rotation gate set and show how to construct versatile ``CNOT units,'' i.e., elementary building blocks, that can be used to parametrize any circuit of a given number of qubits;
    \item We show that optimizing over the parametrized circuit consists of optimizing the structure (i.e., where to place the CNOT units) and optimizing the rotation angles of the rotation gates. We show that the former is largely unimportant once past the so-called surjectivity bound (even when imposing hardware constraints), while the latter is easy from an optimization perspective, by using e.g., Nesterov's accelerated gradient descent~\cite{nesterov1983,beck2009fast};
    \item With the intention to further compress the compiled circuit, allowing for approximation errors in terms of gate fidelity, we propose a novel regularization based on group LASSO~\cite{GLasso2013}, which (among other things) can reduce the length of the QSD down to the CNOT theoretical lower bound without affecting fidelity noticeably (i.e., a factor of two compression);  
    \item Various numerical results support our findings. In particular, we showcase how to use our approach to discover new decompositions in the CNOT$+$rotation gate set of special gates (e.g., the Toffoli gate) and how to use the compression mechanism as an extension of any compilation code available (e.g., Qiskit transpile~\cite{Qiskit2019}) that can trade-off accuracy for circuit \rev{\textit{length}, where length is the number of CNOTs.}  
\end{itemize}

While discussing the main contributions, the paper focuses on three complementary goals,
\begin{itemize}
    \item[\textbf{[G1]}] Approximate quantum compiling for random unitary matrices: here the goal is to derive results in terms of the number of CNOTs to use to compile any given random unitary matrix, with connectivity constraints;
    \item[\textbf{[G2]}] Approximate quantum compiling for special gates: here the goal is to derive results in terms of the number of CNOTs to use to compile special gates;
    \item[\textbf{[G3]}] Approximate quantum compiling for circuit compression: here the goal is to derive results that allow for compression of circuits, freeing the possibility to have inexact compilation. 
\end{itemize}

\textbf{Organization. } This paper is organized as follows. In Section~\ref{sec:preliminaries}, we report the mathematical and physical preliminaries to our algorithmic development. In Section~\ref{sec:math-opt}, we formalize the approximate quantum compiling problem as a mathematical program. In Section~\ref{sec:programmable}, we devise a programmable \cmt{unit} (the two-qubit CNOT unit) with which we can build any circuit. In Section~\ref{sec:propr}, we discuss the property of the mathematical program introduced in Section~\ref{sec:math-opt} when specified for the \cmt{parametric circuit} presented in Section~\ref{sec:programmable}. 

From Section~\ref{sec:struc}, we focus on particular layout patters, namely sequential, spin, and Cartan, and we present their properties. In Section~\ref{sec:gd}, we discuss the use of gradient descent to optimize the rotation angles once the layout is fixed, and we showcase numerical results in Sections~\ref{subsec:random}-\ref{subsec:toffoli}. 
Section~\ref{sec:prune} discusses our proposed compression strategy to trade-off accuracy and \cmt{length}, both theoretically and numerically.
We then conclude in Section~\ref{sec:concl}. 
\section{Preliminaries}
\label{sec:preliminaries}

We work with $n$-qubit quantum circuits, which are represented either by a collection of ordered gate operations, or by a $d$ by $d$ unitary matrix with $d=2^n$. The class of unitary matrices of dimension $d \times d$ together with the operation of matrix multiplication have an important group structure~\cite{knapp2013lie}, denoted as the unitary group $\U(d)$. In particular, the group is a Lie group of dimension $d^2$ as a real manifold; we let $\un(d)$ be its Lie algebra, which consists of anti-Hermitian matrices. Unitary matrices with determinant equal to $1$ are called special unitary matrices. Special unitary matrices of dimension $d \times d$ together with the operation of matrix multiplication form the the special unitary group of degree $d$: $\SU(d)$, which is a Lie group of dimension $d^2-1$ as a real manifold. We let $\su(d)$ denote its Lie algebra, which consists of traceless anti-Hermitian matrices. 

A useful property of the determinant $\det(\cdot)$ of any squared $d$ by $d$ matrix $A$ and scalar $c$ is that $\det(cA)=c^d A$. Hence if $U\in\U(d)$ then $U/\det(U)^{1/d}\in \SU(d)$. Since scalars (e.g., global phases in quantum computing) are easy to implement on a quantum computer (and in fact unimportant), we normalize unitary matrices as above, and so we only need to know how to ``work with'' special unitary matrices. We remark that if $U\in \SU(d)$, then adding any global phase multiple of $2\pi/d$ does not alter the matrix determinant, i.e., $e^{i\frac{2\pi\,m}{d}} U\in\SU(d)$ as well for any integer $m\in \Z$. The equivalence classes from this relation form the projective special unitary group, $\PSU(d)$, which is isomorphic to the projective unitary group, $\PU(d)\coloneqq \U(d)/\U(1)$. Thus, compiling a matrix from $\SU(d)$ also provides a compilation for its $d-1$ equivalent matrices in $\PSU(d)$.

A single-qubit gate on the $j$th qubit is a unitary matrix of the form $I_{2^{j-1}}\otimes u \otimes I_{2^{n-j}}$ where $u\in \U(2)$, $I_q$ is the identity matrix of dimension $q$, and $\otimes$ represents the Kronecker product. An important set of matrices in $\U(2)$ is the set of Pauli matrices, which we denote with their usual notation as $X, Y, Z$, which are unitary, Hermitian, traceless, and have determinant equal to $-1$. From the Pauli matrices, one obtains the rotation matrices $R_x(\theta), R_y(\theta), R_z(\theta)$ by matrix exponentiation. The rotation matrices are special unitary. 
An important fact that we use extensively in this paper is that any $u\in \SU(2)$ can be written as a product of any three rotation matrices with no two consecutive the same~\cite{nielsenbook}. 

Among multiple-qubit gates, we focus on the two-qubit CNOT gate, or controlled-$X$ gate, with control qubit $j$ and target qubit $k$, which is the matrix
\begin{align*}
    &\textrm{CNOT}_{jk} =
    I_{2^{j-1}}\otimes \begin{bmatrix}
    1 & 0\\
    0 & 0
    \end{bmatrix}
    \otimes I_{2^{n-j}}\\
    &\hspace{2cm}+\begin{cases}
    I_{2^{j-1}}\otimes \begin{bmatrix}
    0 & 0\\
    0 & 1
    \end{bmatrix}
    \otimes I_{2^{k-j-1}}\otimes X\otimes I_{2^{n-k}} & \text{ if } j<k\\
    I_{2^{k-1}}\otimes X\otimes I_{2^{j-k-1}} \otimes \begin{bmatrix}
    0 & 0\\
    0 & 1
    \end{bmatrix}
    \otimes I_{2^{n-j}} & \text{ if } k<j
    \end{cases}.
\end{align*}
For example, for $n=2$, we have
\begin{align*}
    \textrm{CNOT}_{12} = \begin{bmatrix}
    1&0&0&0\\0&1&0&0\\0&0&0&1\\0&0&1&0
    \end{bmatrix}, \qquad
    \textrm{CNOT}_{21} = \begin{bmatrix}
    1&0&0&0\\0&0&0&1\\0&0&1&0\\0&1&0&0
    \end{bmatrix}.
\end{align*}
Note that for $n=2$, CNOT has determinant $-1$, and so we have to normalize it. On the other hand, for $n>2$, CNOT has determinant 1.

\section{Approximate quantum compiling as mathematical optimization}
\label{sec:math-opt}

We are interested in compiling a quantum circuit, which we formalize as finding the ``best'' circuit representation in terms of an ordered gate sequence of a target unitary matrix $U\in \U(d)$, with some additional hardware constraints. In particular, we look at representations that could be constrained in terms of hardware connectivity, as well as \cmt{circuit length}, and we choose a gate basis in terms of CNOT and rotation gates. The latter choice is motivated by an implementation in the Qiskit software package~\cite{Qiskit2019}. We recall that the combination of CNOT and rotation gates is universal in $\SU(d)$ and therefore it does not limit compilation~\cite{nielsenbook}.

To properly define what we mean by ``best'' circuit representation, we define the metric as the Frobenius norm between the unitary matrix of the compiled circuit $V$ and the target unitary matrix $U$, i.e., $\|V - U\|_{\mathrm{F}}$. This choice is motivated by mathematical programming considerations, and it is related to other formulations that appear in the literature (see Remark~\ref{rem:fid}). 

We are now ready to formalize the approximate quantum compiling problem as follows. 

\emph{Given a target special unitary matrix $U \in \SU(2^n)$ and a set of constraints, in terms of connectivity and \cmt{length}, find the closest special unitary matrix $V \in \mathcal{V} \subseteq \SU(2^n)$, where $\mathcal{V} $ represents the set of special unitary matrices that can be realized with rotations and CNOT gates alone and satisfy both connectivity and \cmt{length} constraints, by solving the following mathematical program: }
\begin{equation}\label{eq:aqcp}
\textbf{(AQCP)}\qquad \min_{V \in \mathcal{V} \subseteq SU(2^n)} \, f(V):=\frac{1}{2}\|V - U\|_{\mathrm{F}}^2.
\end{equation}

\rev{Note that the cost function can be equivalently written
\begin{align}
    \frac{1}{2}\|V-U\|_F^2 &= \frac{1}{2}\tr[U\dagg U-U\dagg V-V\dagg U+V\dagg V] = \frac{1}{2}\tr[2I]-\frac{1}{2}\tr[U\dagg V]-\frac{1}{2}\tr[V\dagg U]\notag\\
    &= d-\Re~\tr[U\dagg V]\label{eq:retr}
\end{align}
where we used that $U$ and $V$ are unitary matrices.}

We call~\eqref{eq:aqcp} the approximate quantum compiling (master) problem (AQCP). A solution of the problem is an optimal $V$ indicated as $V^*$, along with an ordered set of gate operations that respect the constraints.  

If $V^*$ is such that the cost is null, then we say that the compilation is exact, otherwise it is approximate and the approximation error is computed by the cost $\frac{1}{2}\|V^* - U\|_{\mathrm{F}}^2$.  Without further specifications (and simplifications), Problem~\eqref{eq:aqcp} is intractable but for very small circuits. How to efficiently solve~\eqref{eq:aqcp} is our aim. 


\begin{remark}\label{rem:fid}
In~\cite{khatri2019}, a different metric was used, namely the Hilbert-Schmidt test defined as 
$$
C_{HST}(U,V) = 1 - \frac{1}{d^2} |\tr[V^\dag U]|^2 = \frac{d+1}{d}(1 - \bar{F}(U,V)), 
$$
where $\bar{F}(U,V)$ is the fidelity averaged over the Haar distribution. Given that $f(V) = \frac{1}{2}\|V - U\|_{\mathrm{F}}^2 = d - \Re \tr[V^\dag U]$, then $\Re \tr[V^\dag U] = d - f(V)$. Hence,
\begin{multline*}
\bar{F}(U,V) = 1 - \frac{d}{d+1}  C_{HST}(U,V) = 1- \frac{d}{d+1} + \frac{1}{d(d+1)} ((\Re \tr[V^\dag U])^2 + (\Im \tr[V^\dag U] )^2) \geq \\  1- \frac{d}{d+1} + \frac{1}{d(d+1)} (d - f(V))^2 =: \bar{F}_{\mathrm{F}}(U,V),
\end{multline*}
where we have defined  the new quantity $\bar{F}_{\mathrm{F}}(U,V)$ as the Frobenius fidelity, which is always lower than the $\bar{F}(U,V)$. By minorization arguments, minimizing $f(V)$ has the effect of maximizing the \rev{Frobenius} fidelity.
\end{remark}

The stepping stones of our work are the papers~\cite{cincio2018,khatri2019}. There, the AQCP is approached by a bi-level technique. Since $V$ needs to be sought in the space of matrices that can be realized with rotation and CNOT gates, one can solve for $V$ by interleaving an optimization on the structure (at fixed \cmt{length}), i.e., which gate needs to be implemented where, and an optimization on the rotation angles. The first problem (deciding the structure) is combinatorial and non-convex, and \cmt{\cite{cincio2018,khatri2019}} propose a meta-heuristic based on simulated-annealing and compactifications, while the second problem (deciding the rotation angles) is continuous and non-convex, and \cmt{\cite{cincio2018,khatri2019} include} a gradient descent algorithm. \cmt{\cite{cincio2018,khatri2019}} also offer an algorithm to fix the CNOT structure, more or less arbitrarily, and optimize only for rotation angles, which is more efficient in terms of computation time, but less optimal. 

Despite the encouraging results, key challenges in these works remain. First, optimizing the structure via simulated-annealing is far from optimal and it can be very time consuming. But most importantly it is physics-agnostic, therefore one may wonder if one could do better by using physics. 

Second, fixing the structure, as is done in~\cite{khatri2019}, is by no means universal nor does it mirror the hardware connectivity; a better structure could be a sequential structure (see Figure~\ref{fig:str-0}).

\begin{figure}
\centering
\includegraphics[width=0.7\textwidth]{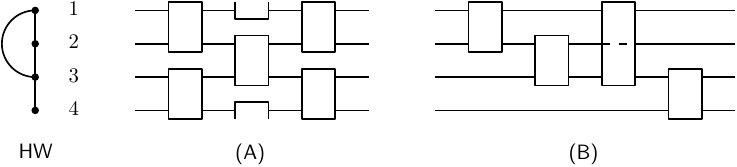}
\caption{Possible two-qubit gate units structures. HW depicts the hardware connectivity of a four-qubit quantum computer. (A) the structure, or ansatz, proposed in~\cite{khatri2019}, where each block represents a CNOT gate with rotation gates before and after, is easily implementable but may miss some hardware connections or introduce some that are not there; (B) a possible sequential structure, which alternates among all the possible two-qubit blocks and could capture hardware connectivity better.}
\label{fig:str-0}
\end{figure}

Third, the gradient descent algorithm used in~\cite{khatri2019} seems to work surprisingly well, despite the non-convexity of the problem at hand, and one may wonder why this is so, which is a non-trivial question. 

The rest of this paper answers these questions by using a pertinent parametrization of the matrix $V$ in terms of CNOTs and rotation gates. 

\section{A programmable two-qubit gate and the CNOT unit}
\label{sec:programmable}

We start by defining a flexible two-qubit gate, which is parametrized by four rotation angles, and which \cmt{will} be the building block of \cmt{the parametric circuit}. 

Let us focus for now on basis set constraints alone. Given the target special unitary $U\in\SU(d)$, compiling means to find a product of matrices from the basis set (as said CNOT$+$rotations) that either approximates or equals $U$. We define the \cmt{``length''} of the compilation as the number of CNOTs in the product. In 2004, Shende, Markov, and Bullock derived a universal lower bound on the length for exact compilation in this setting~\cite{shende2004lowerbound}. The bound specifies a \emph{sufficient} minimal number of CNOTs for exact compilation, for \emph{all} the matrices in $\SU(d)$. And, while a particular special unitary matrix may be exactly compiled with a smaller length than the universal lower bound (the bound is not \emph{necessary}), there exists a special unitary matrix that cannot (the bound is sufficient and necessary for at least one special unitary matrix). Moreover, the set of special unitary matrices that can be compiled with a smaller length than the universal lower bound has measure zero in $\SU(d)$. 

We summarize the proof of this fact, as it motivates our parametrization ideas.

\begin{lemma}{\cite[Prop. III.1]{shende2004lowerbound}} The set of special unitary matrices $U \in \SU(d)$, $d = 2^n$, that do not need at least $\big\lceil \frac{1}{4}(4^n-3n-1)\big\rceil$ CNOTs to be exactly compiled has measure zero in $\SU(d)$.
\end{lemma}

\begin{proof}
First, since each rotation gate has $1$ real parameter while $\SU(d)$ has real dimension $d^2-1$, Sard's theorem~\cite{lee2013smooth} can be used to show that the set of special unitary matrices that do not need $d^2-1$ rotation gates to be exactly compiled has measure zero in $\SU(d)$ \cite[Lem. II.2]{shende2004lowerbound}. Thus, if a product of matrices from the basis set can be reduced to a product with less than $d^2-1$ rotation gates, then it is in the measure zero set. In particular, more than $3$ consecutive rotation gates on the same qubit in a product can be reduced to only $3$ rotation gates. So, without CNOTs, any product can be reduced to one with only $3n<d^2-1$ rotation gates. Furthermore, $R_z$ commutes with the control of CNOTs and $R_x$ commutes with the target of CNOTs. Thus, whenever $3$ rotation gates are applied after a CNOT on either the control or target qubit, we can rewrite them as $3$ rotation gates such that one of the rotation gates commutes with the CNOT. Thus, we can reduce any product of matrices from the basis set to a product with $3$ rotation gates applied to each qubit, followed by CNOTs, each followed by only $4$ rotation gates. Thus, a product with $L$ CNOTs can be reduced to a product with $4L+3n$ rotation gates. So, the set of special unitary matrices that do not need at least $\big\lceil \frac{1}{4}(4^n-3n-1)\big\rceil$ CNOTs to be exactly compiled has measure zero in $\SU(d)$ \cite[Prop. III.1]{shende2004lowerbound}.
\end{proof}

Therefore, for an arbitrary $n$-qubit circuit, one would need at least $\Omega(4^n)$ CNOT gates for exact compilation. Henceforth, we refer to the bound $\big\lceil \frac{1}{4}(4^n-3n-1)\big\rceil$ as the theoretical lower bound, or TLB for short. Figure~\ref{fig.expl} gives a glimpse \cmt{of} how the reduction in terms of CNOTs and rotations can be carried out for a 3-qubit circuit.

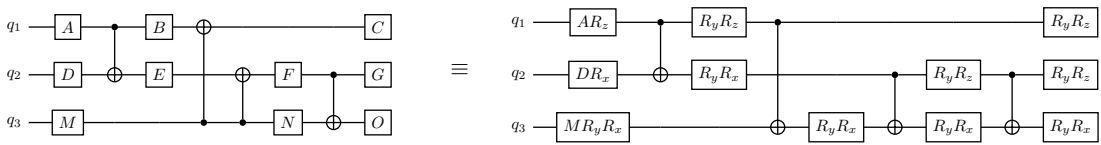
\begin{figure}[th]
\centering
\begin{tikzpicture}
\node[scale=.62] {
\begin{quantikz}
\lstick{$q_1$} & 
 \gate{A} & \ctrl{1} & \gate{B} & \targ{} & \qw & \qw & \qw & \gate{C} & 
 \rstick{} \\
\lstick{$q_2$} & 
 \gate{D} & \targ{} & \gate{E} & \qw & \targ{} & \gate{F} & \ctrl{1} & \gate{G} &
\rstick{}\\
\lstick{$q_3$} & 
 \gate{M} & \qw & \qw &  \ctrl{-2} & \ctrl{-1} & \gate{N} & \targ{} & \gate{O} &
\rstick{}
\end{quantikz} \Large $\quad \equiv \quad$ \normalsize
\begin{quantikz}
\lstick{$q_1$} & 
 \gate{A \rev{R_z}} & \ctrl{1} & \gate{R_y R_z} & \ctrl{2} & \qw & \qw & \qw & \qw & \gate{R_y R_z} & 
 \rstick{} \\
\lstick{$q_2$} & 
 \gate{D \rev{R_x}} & \targ{} & \gate{R_y R_x} & \qw& \qw & \ctrl{1} & \gate{R_y R_z} & \ctrl{1} & \gate{R_y R_z} &
\rstick{}\\
\lstick{$q_3$} & 
 \gate{M \rev{R_y R_x}} & \qw & \qw &  \targ{} &\gate{R_y R_x} &\targ{} & \gate{R_y R_x} & \targ{} & \gate{R_y R_x} &
\rstick{}
\end{quantikz}};
\end{tikzpicture}
\caption{An example of an equivalent quantum circuit with CNOT and rotations gates, all the capital letter gates, e.g. $A$, are single-qubit unitaries that can be compiled via three rotation gates up to an irrelevant global phase. \rev{Note that the $R_y$ next to $M$ results from flipping the two upward-facing CNOTs and has angle $-\pi/2$.}}
\label{fig.expl}
\end{figure}

While \cite{shende2004lowerbound} only proved a lower bound, their reasoning motivates \cmt{a flexible parametric circuit construction}. Since every product of matrices from the basis set can be reduced to a product of the form described in the previous paragraph (see also Figure~\ref{fig.expl}), we will consider products of that form with different sequences of CNOTs. Let $\CU_{j\to k}(\theta_1,\theta_2,\theta_3,\theta_4)$ be the matrix represented by 

\begin{center}
\begin{quantikz}
\lstick{$j$} & \ctrl{1} & \gate{R_y(\theta_1)} & \gate{R_z(\theta_2)} & \qw \\ \lstick{$k$} & \targ{} & \gate{R_y(\theta_3)} & \gate{R_x(\theta_4)} & \qw & \rstick{.}
\end{quantikz}
\end{center}

We call this a ``CNOT unit'' and we  use it as a programmable two-qubit block in \cmt{the parametric circuit}. Define
$$
J(L)=\{(j_1\to k_1,...,j_L\to k_L)\mid j_m<k_m\in \{1,\ldots,n\}\cmt{\text{ for all }}m\in\{1,\ldots,L\}\},
$$ 
as the set of $L$-long lists of pairs of control-target indices with the control index smaller than the target index. The latter constraint is without loss of generality since a CNOT can be flipped with $y$-rotations \rev{($R_y(\pi/2)$ on the left of the control and right of the target, and $R_y(-\pi/2)$ on the right of the control and left of the target)}. 
The set $J(L)$ represents all the possible CNOT unit sequences that one can have of length $L$. In Figure~\ref{fig.expl}, we have one element of $J(4)$ as $(1\to2,1\to3,2\to3,2\to3)$.

The cardinality of $J(L)$ can be shown to be $|J(L)|=(n(n-1)/2)^L$. Given an element $\ct\in J(L)$ (that is, an $L$-long list of pairs defining the location of our CNOT units, for example in Figure~\ref{fig.expl} $\ct = (1\to2,1\to3,2\to3,2\to3)$) we define $\ct(i)$ as the pair at position $i$, and we let $\Vct:[0,2\pi)^{3n+4L}\to\SU(2^n)$ be the following circuit
\begin{align*}
    \Vct(\btheta)=& \CU_{\ct(L)}(\theta_{3n+4L-3},...,\theta_{3n+4L})\cdots \CU_{\ct(1)}(\theta_{3n+1},...,\theta_{3n+4})\\
    &[R_z(\theta_{1})R_y(\theta_{2})R_z(\theta_{3})]\otimes\cdots\otimes [R_z(\theta_{3n-2})R_y(\theta_{3n-1})R_z(\theta_{3n})],
\end{align*}
where we have collected all the angles $(\theta_1, \ldots, \theta_{3n+4L})$ into the vector $\btheta$.  The circuit $\Vct$ corresponds to the reduced product in the lower bound proof and will be the main object of our study: it will be the circuit blueprint for any circuit realizable in hardware (see also Figure~\ref{fig:str-1}).

\begin{figure}
\centering
\includegraphics[width=0.6\textwidth]{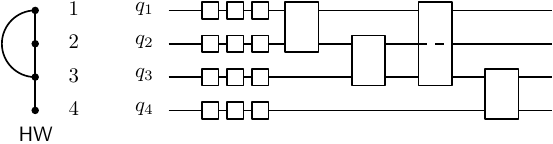}
\caption{Decomposition of a realizable quantum circuit in terms of $\Vct(\btheta)$. The first three one-qubit gates for all qubits are rotation gates, $R_z, R_y, R_z$, while the two-qubit blocks represents CNOT units as specified. Here $L = 4$, \cmt{hence there are} $28$ rotation angles variables. \cmt{Moreover,} $\ct = (1\to 2,2\to 3,1\to 3,3\to 4)$, \cmt{so} the hardware connectivity \cmt{is satisfied}.}
\label{fig:str-1}
\end{figure}


\begin{remark}\label{rem:qsd}
A useful circuit decomposition is the quantum Shannon decomposition (QSD), proposed in 2006 by
Shende, Bullock, and Markov, which uses $\frac{23}{48}4^n-\frac{3}{2}2^n+\frac{4}{3}$ CNOTs to decompose any $n$-qubit circuit \cite{shende2006}. This number of CNOTs is \cmt{only twice the TLB.} The resulting circuit is formalized as a sequence of recursive Cartan decompositions in \cite{drury2008constructive}, and Theorems 4, 8, and 12 in \cite{shende2006} can be recursively used to find an explicit $\ct$ \cmt{corresponding to the QSD}. For $n=3$, $\ct=(1\to 2, 1\to 2, 1\to 2, 2\to 3, 1\to 3, 2\to 3, 1\to 2, 1\to 2, 1\to 2, 2\to 3, 1\to 3, 2\to 3, 1\to 3, 1\to 2, 1\to 2, 1\to 2, 2\to 3, 1\to 3, 2\to 3, 1\to 2, 1\to 2, 1\to 2)$ \cite[Fig. 3]{drury2008constructive}. 
\end{remark}

At this point, we can already reformulate the approximate quantum compiling problem~\eqref{eq:aqcp}, in the equivalent form
\begin{equation}\label{eq:aqcp1}
\textbf{(AQCP-CT)}\qquad \min_{L \in [1,\ldots,\bar{L}], \btheta\in [0, 2\pi)^{3n+4L}, \ct \in \mathcal{C}(L)} \, f_{\ct}(\btheta):=\frac{1}{2}\|\Vct(\btheta) - U\|_{\mathrm{F}}^2,
\end{equation}
where the set $\mathcal{C}(L)$ represents the set of \rev{$L$-long} realizable lists in hardware (for connectivity limitations), and $\bar{L}$ is the \cmt{length} limit. 

Problem~\eqref{eq:aqcp1} is by no means easier than the original Problem~\eqref{eq:aqcp}. However, it is instructive to understand its properties and this will help us devise better solution strategies. For instance, if we were to drop hardware constraints and adopt the QSD strategy of Remark~\ref{rem:qsd}, then Problem~\eqref{eq:aqcp1} could simplify into:
\begin{equation}\label{eq:aqcpqsd}
\textbf{(AQCP-QSD)}\qquad \min_{L = \frac{23}{48}4^n-\frac{3}{2}2^n+\frac{4}{3}, \btheta\in [0, 2\pi)^{3n+4L}, \ct = \mathrm{QSD}(n)} \, f_{\ct}(\btheta):=\frac{1}{2}\|\Vct(\btheta) - U\|_{\mathrm{F}}^2,
\end{equation}
which is a non-convex (but continuous) optimization problem in the $\btheta$ variables. In fact, since the structure \cmt{is fixed,} the only free parameters are the rotation angles. But on the other hand, can we do better in terms of number of CNOTs and in terms of incorporating the hardware constraints?

\section{The optimization landscape}
\label{sec:propr}

We start by deriving some useful properties of the cost function $f_{\ct}(\btheta):=\frac{1}{2}\|\Vct(\btheta) - U\|_{\mathrm{F}}^2$ given the target unitary $U$ and a fixed structure. In particular, we will look at $\Vct$ and \cmt{its} properties of differentiability, \cmt{as well as the} first and second order derivatives \cmt{of $f_{\ct}$}, which are important for optimization purposes. 

First, we would like to prove that $\Vct(\btheta)$ is a smooth mapping in the rotation angles $\btheta$. To do so, we apply the partial derivative operator, $\partialk$, to $\Vct$, with respect to an arbitrary angle $\theta_k$, where $k$ is an index in $[1, 3n+4L]$. We further let $R_{g_k}$ be the rotation gate associated to that angle and $\sigma_{g_k}$ be the respective Pauli operator ($g$ can be $x,y,z$ depending on which type of rotation gate the angle refers to). With this notation, $\frac{d}{d\theta_k}R_{g_k}(\theta_k)=-\frac{i}{2}\sigma_{g_k} R_{g_k}(\theta_k)$, and $\frac{d^2}{d\theta_k^2}R_{g_k}(\theta_k)=-\frac{1}{4} R_{g_k}(\theta_k)$. Thus, if a sequence of partial derivative operators are applied to $\Vct$, to compute its gradient, or Hessian, or higher order derivative, the result is a complex constant times $\Vct$ with Pauli gates inserted at specific locations. Thus, we have proved the following theorem.
\begin{theorem}
\label{thm:smooth}
The circuit $\Vct(\btheta)$ is infinitely differentiable with respect to the rotation angles $\btheta$.
\end{theorem}

Next, we look at the cost $f_{\ct}(\btheta)$. \rev{Using Eq.~\eqref{eq:retr}, we get}
\begin{align*}
    \partialk \frac{1}{2}\|\Vct(\btheta)-U\|_F^2&= -\Re~\tr\left[\partialk \Vct(\btheta)\dagg U\right]
\end{align*}
and
\begin{align}\label{eq:elements}
    \frac{\partial^2}{\partial\theta_k\partial\theta_{\ell}} \frac{1}{2}\|\Vct(\btheta)-U\|_F^2&= -\Re~\tr\left[\frac{\partial^2}{\partial\theta_k\partial\theta_{\ell}} \Vct(\btheta)\dagg U\right].
\end{align}
In particular, by the fact that $\frac{d^2}{d\theta_k^2}R_{g_k}(\theta_k)=-\frac{1}{4} R_{g_k}(\theta_k)$, for the diagonal elements:
\begin{align}\label{eq:diag}
    \frac{\partial^2}{\partial\theta_k^2} \frac{1}{2}\|\Vct(\btheta)-U\|_F^2&= \frac{1}{4}\Re~\tr[ \Vct(\btheta)\dagg U] = \frac{d}{4}-\frac{1}{8}\|\Vct(\btheta)-U\|_F^2.
\end{align}
Now we are ready to show that $f_{\ct}(\btheta)$ is strongly smooth in the optimization sense (i.e., its gradient is Lipschitz continuous).

\begin{theorem}
\label{thm:hessian}
\cmt{For all $\btheta$, the Hessian of $f_{\ct}(\cdot)=\frac{1}{2}\|\Vct(\cdot)-U\|_F^2$ evaluated at $\btheta$} has spectrum in $[-(3n+4L-3/4)d,(3n+4L-3/4)d]$.
\end{theorem}
\begin{proof} First of all, the Hessian is real and symmetric, so its eigenvalues lie on the real line. Then, note that for all unitary matrices, $W\in\U(d)$: $\|W\|_{\textrm{F}}^2=d$ and the farthest unitary matrix from any $W$ is $-W$. Thus, $\Re~\tr[\cdot\dagg \cdot]:\U(d)^2\to[-d,d]$. Hence, using~\eqref{eq:elements}, we can show that each \cmt{off}-diagonal element of the $(3n+4L)$ by $(3n+4L)$ Hessian matrix is bounded in absolute value by $d$. As for the diagonal elements, Eq.~\eqref{eq:diag} says that they are bounded below by $-d/4$ and above by $d/4$. 

With this in place, we can use Gershgorin's disc theorem~\cite[Theorem 6.1.1.]{horn_matrix_2012} with centers in $[-d/4,d/4]$ and radii in $[\cmt{0},(3n+4L-1)d]$ to prove the claim.
\end{proof}

Theorem \ref{thm:hessian} says that $f_{\ct}(\btheta)$ is a \cmt{strongly} smooth function (since the \cmt{spectrum of the} Hessian is \cmt{uniformly} bounded), which implies fast convergence to a stationary point (i.e., points for which the gradient vanishes) for gradient descent. In addition, even though $f_{\ct}(\btheta)$ is non-convex, gradient descent with random initialization \cite{lee2016gradient}, perturbed gradient descent \cite{jin2017escape}, and perturbed Nesterov's method \cite{jin2018accelerated} all converge to second-order stationary points under \cmt{strong} smoothness. Second-order stationary points are stationary points where the Hessian is positive semi-definite, and therefore are local minima. So, applying these methods to $f_{\ct}(\btheta)$, we will be sure to reach at least a local minimum.

In some cases, such as in non-convex low rank problems, most second-order stationary points are, in fact, global minima \cite{ge17no}. Even though this is not the case for $f_{\ct}$, we do not need to converge to a global minimum necessarily. We would be content finding a $\btheta^*$ such that $V_{\ct}(\btheta^*)$ is a global minimum up to a global phase transformation. 

To this aim, we now present supporting facts that lead us to conjecture that, under certain conditions, we can easily find global minima up to a global phase transformation.

First, we report a technical lemma. 
\begin{lemma}
\label{lem:orthog}
Indicate with $\orthog$ the orthogonal complement with respect to the $\Re\tr[(\cdot)^\dag (\cdot)]$ operation. The orthogonal complement $\un(d)\orthog$ is the set of Hermitian matrices and $\su(d)\orthog=\un(d)\orthog+\spn\{iI\}$.
\end{lemma}
\begin{proof}
Recall that $\un(d)$ consists of anti-Hermitian matrices and $\su(d)$ consists of traceless anti-Hermitian matrices. The orthogonal complement $\un(d)\orthog$ is the set of all the matrices $H$ for which $\Re\tr[H^\dag U]=0$, for all $U$ anti-Hermitian. By direct calculation, indicating with $h_{ij}$ the element $i,j$ of matrix $H$, and with $u_{i,j}$ the one of matrix $U$, as well as $\bar{h}_{i,j}$ the conjugate of $h_{i,j}$, then,
\begin{equation*}
   \Re\tr[H^\dag U] = \Re \left[\sum_{i}\sum_{j} \bar{h}_{i,j} u_{i,j} \right]. 
\end{equation*}
If the above has to be $0$ for all anti-Hermitian matrices $U$, then $H$ has to be Hermitian by $u_{j,i} = -\bar{u}_{i,j}$ (in particular, the diagonal of $U$ is only imaginary). On the other hand, if $H$ is Hermitian (therefore its diagonal is only real), then,
\begin{equation*}
   \Re\tr[H^\dag U] = \underbrace{\Re \Big[\sum_{j>i} \bar{h}_{i,j} u_{i,j} + \sum_{j<i} {h}_{i,j} \bar{u}_{i,j}\Big]}_{=0} + \underbrace{\Re\Big[\sum_{i=j} \bar{h}_{i,j} u_{i,j}  \Big]}_{=0} = 0
\end{equation*}
Hence, the orthogonal complement $\un(d)\orthog$ is the set of Hermitian matrices.

In addition, $\su(d)$ consists of traceless anti-Hermitian matrices, so its orthogonal complement consists not only of Hermitian matrices, but also of any matrix of the form $H+c i I$, where $H$ is Hermitian and $c \in \mathbb{R}$. Thus, the thesis follows.
%
%
\end{proof}

With the Lemma in place, we are ready for another intermediate result that pertains to stationary points. 

\begin{theorem}
\label{thm:stationary}
\cmt{Let $U \in \SU(d)$ be a special unitary matrix, and $V(\cdot):\R^p\to \SU(d)$ be a function that is infinitely differentiable, surjective, and of constant rank. \cmt{For all} $U\in \SU(d)$, a parameter value $\btheta \in [0, 2\pi)^p$ is a stationary point of the cost $f(\cdot)=\frac{1}{2}\|V(\cdot)-U\|_{\mathrm{F}}^2$ if and only if $V(\btheta)\dagg U$ has eigenvalues in $\{e^{i\alpha},-e^{-i\alpha}\}$ for some $\alpha\in [0,2\pi)$.}
\end{theorem}
\begin{proof}
Assume the hypotheses. We have the following line of implications:
\begin{align*}
    \btheta\text{ is a stationary point of }f \hspace{.5cm}&\overset{1}{\iff} \Re~\tr\bigg[\partialk V(\btheta)\dagg U\bigg]=0\cmt{\text{ for all }} k\in[p]\\
    &\overset{2}{\iff} U\in \spn\bigg\{\partialk V(\btheta)\bigg\}\orthog\\
    &\hspace{1.7cm}\overset{3}{=}\left[\mathcal{T}_{V(\btheta)}\SU(d)\right]\orthog\\
    &\hspace{1.7cm}\overset{4}{=}\left[V(\btheta)\su(d)\right]\orthog\\
    &\hspace{1.7cm}\overset{5}{=}V(\btheta)\su(d)\orthog\\
    &\iff V(\btheta)\dagg U\in \su(d)\orthog.
\end{align*}

Step 1 follows from the definition of \cmt{a} stationary point (a point where the gradient vanishes) and our derivation of the gradient for $f(\cdot)=\frac{1}{2}\|V(\cdot)-U\|_{\mathrm{F}}^2$. Step 2 follows from the definition of \cmt{an} orthogonal complement and the linearity of the trace (i.e., $U$ is in the span of the orthogonal complement of the derivative of $V$ \cmt{if and only if} the gradient is zero). Step 3 follows from the global rank theorem \cite[Thm 4.14]{lee2013smooth} under infinite differentiability, surjectivity, and constant rank assumptions, where we use $\mathcal{T}$ to denote the tangent space (which appears since we are taking the derivative of \cmt{$V(\btheta)$}). Concretely, given $W\in \SU(d)$, $\mathcal{T}_W \SU(d)\coloneqq \{\gamma'(0)\mid \gamma:\R\to\SU(d)\text{ smooth with }\gamma(0)=W\}$. For Step 4, a little more care has to be put. First, it can be shown, via the left or right group action, that $\mathcal{T}_W \SU(d)$ is isomorphic to $\mathcal{T}_I \SU(d)=\su(d)$. We want now to show that $\mathcal{T}_W \SU(d)=W\su(d)$. Since they are isomorphic, we only have to show one direction. Let $A\in\su(d)$. We want to show $WA\in \mathcal{T}_W \SU(d)$. Towards this end, define $\gamma(t)=W\exp(tA)$. Note $\gamma(0)=W$, and $\gamma'(0) = WA$. Furthermore, $\gamma(t) \in \SU(d)$ (which is easy to see by direct computations\footnote{
We have $
    \gamma(t)\dagg\gamma(t) = \exp(tA)\dagg W\dagg W\exp(tA) = \exp(tA)\dagg \exp(tA)= \exp(tA\dagg)\exp(tA)= \exp(-tA)\exp(tA)= \exp(-tA+tA) = I
    $
and
$\det(\gamma(t)) = \det(W\exp(tA) =\det(W)\det(\exp(tA))= \det(\exp(tA))
= \exp(t\tr(A))= \exp(0)=1.
$}).
Thus, $WA=\gamma'(0)\in \mathcal{T}_W \SU(d)$ and so $\mathcal{T}_W \SU(d)=W\su(d)$, proving Step 4. Step 5 follows via properties of the Hermitian transpose. Explicitly, for a set $\mathcal{U}$,
\begin{align*}
    [W\mathcal{U}]\orthog &= \{Y\mid \Re~\tr [Y\dagg WA]=0\cmt{\text{ for all }} A\in\mathcal{U}\}\\
    &=\{Y\mid \Re~\tr [(W\dagg Y)\dagg A]=0\cmt{\text{ for all }} A\in\mathcal{U}\}\\
    &= W\{W\dagg Y\mid \Re~\tr [(W\dagg Y)\dagg A]=0\cmt{\text{ for all }} A\in\mathcal{U}\}\\
    &= W\mathcal{U}\orthog.
\end{align*}
Now all that is left is to determine the contents of $\su(d)\orthog\cap \U(d)$. We have from Lemma \ref{lem:orthog} that $\su(d)\orthog$ is the set of matrices that equal $A+ciI$ for some Hermitian $A$ and some $c\in\R$. Furthermore, Hermitian matrices are precisely those that equal $Q\dagg D Q$ for some unitary $Q$ and diagonal $D$ with real entries. So, $\su(d)\orthog$ is the set of matrices that are unitarily diagonalizable with eigenvalues in \rev{$\{\lambda+ci\mid \lambda\in\R\}$} for some $c\in\R$. But, unitary matrices are precisely the matrices that are unitarily diagonalizable with eigenvalues in $\U(1)$. Thus, $\su(d)\orthog\cap \U(d)$ is the set of matrices that are unitarily diagonalizable with eigenvalues in $\{e^{i\alpha},-e^{-i\alpha}\}$ for some $\alpha$, since \rev{$\{\lambda+ci\mid \lambda\in\R\}\cap \U(1)=\{e^{i\alpha},-e^{-i\alpha}\}$} for $\alpha=\arcsin(c)$.
\end{proof}

Theorem~\ref{thm:stationary} describes the space of stationary points for infinitely differentiable, surjective, constant rank mappings. Note that the image of points that do not have constant rank has measure zero in $\im(V)=\SU(d)$ by Sard's theorem \cite[Ch. 6]{lee2013smooth}, so this assumption is not restrictive. Moreover, we have from Theorem~\ref{thm:smooth} that $\Vct$ is infinitely differentiable, so we already have a relevant mapping that is infinitely differentiable. Finally, the surjectivity assumption appears as though it could be relaxed. But, unfortunately, this is not the case.



\begin{remark}
\label{rmk:surj}
One might wonder if the surjectivity assumption can be relaxed to the assumption $U\in\im(V)$. The answer is no. \cmt{We found counter-examples when we were running the numerical experiments. In particular, for certain structures $\ct$ with length smaller than the TLB, hence lacking surjectivity, and with the Toffoli gate as the target unitary $U$, we computed some stationary points of $f_{\ct}$ that exactly compiled $U$ and some that failed the eigenvalue condition of Theorem~\ref{thm:stationary}. In other words, for these examples, $U\in\im(\Vct)$ but the conclusion of Theorem 3 does not hold. These results are shown in Table~\ref{tab:my_label}: the stationary points in the unsuccessful compilations did not satisfy the eigenvalue condition of Theorem~\ref{thm:stationary}.}
\end{remark}

Since $\Vct$ is infinitely differentiable, if only it is surjective as well then its stationary points are precisely the points such that $\Vct(\btheta)\dagg U$ has eigenvalues in $\{e^{i\alpha},-e^{-i\alpha}\}$ for some $\alpha\in [0,2\pi)$. Having more than one eigenvalue means that the stationary points are not in general global minima up to a global phase transformation though. But, we conjecture that the form of stationary points in Theorem \ref{thm:stationary} simplifies even further for second-order stationary points of $\Vct$.

\begin{conjecture}
\label{conj:stat}
\cmt{If the parametric circuit $\Vct(\cdot)$ is surjective and the Hessian of the cost function $f_{\ct}(\cdot)=\frac{1}{2}\|\Vct(\cdot)-U\|_{\mathrm{F}}^2$ is positive semi-definite at a stationary point $\btheta$, then $\Vct(\btheta)\dagg U=e^{i\alpha}I$ for some $\alpha\in[0,2\pi)$, meaning that any stationary point $\btheta$ is a global minimum of the cost function $f_{\ct}(\cdot)$, and therefore a solution to the quantum compiling problem, up to a global phase.}
\end{conjecture}

\cmt{We tested Conjecture~\ref{conj:stat} extensively. In order to falsify the conjecture, we would need an example of $\Vct$, $\btheta$, and $U$ such that $\Vct$ is surjective, the Hessian is positive semi-definite, and $\Vct(\btheta)\dagg U$ is not a scalar matrix. We searched for counter-examples in two experiments. In both, we used structures with length smaller than the TLB as the non-surjective mappings, and the QSD decomposition structure as the surjective mapping. In the first experiment, we computed $\btheta$ via gradient descent. Hence, there was no need to explicitly compute the Hessian. In the second experiment, we randomly generated $W\in \su(d)\orthog\cap \SU(d)$ and $\btheta\in \R^p$, then set $U=\Vct(\btheta)W$. From the proof of Theorem~\ref{thm:stationary}, we have that $\su(d)\orthog\cap \SU(d)$ is the set of matrices that are unitarily diagonalizable with eigenvalues in $\{e^{-\alpha},-e^{-i\alpha}\}$ for some $\alpha$ such that the multiplicity of $e^{i\alpha}$ times $\alpha$ plus the multiplicity of $-e^{-i\alpha}$ times $\pi-\alpha$ is an integer multiple of $2\pi$. So, in order to generate $W$, we sampled $Q\in U(d)$ randomly and chose \rev{a non-scalar} eigenvalue matrix $\Lambda$, then set $W=Q\dagg \Lambda Q$. \rev{Thus, $\Vct(\btheta)\dagg U$ was a non-scalar matrix by construction.} We did not find any counter-examples in either experiment. Note that we did have to explicitly compute the Hessian and its eigendecomposition for the second experiment, but this is the only place in the paper where we actually compute the Hessian.}

Conjecture \ref{conj:stat} says that, despite the difficulty in finding global minima due to non-convexity, the stationary points we find are global minima up to a global phase. However, this is only in the case of surjectivity, and so the non-convexity causes greater difficulty for lengths shorter than the minimum required for surjectivity. We discuss this issue further in the next section. In particular, the divide between easy compilation problems and harder ones is the surjectivity of $\Vct(\btheta)$. 

\section{Special structures}
\label{sec:struc}

While the results and discussion in the previous section were for arbitrary CNOT unit structures, we would like to focus on particular structures henceforth. The properties that we desire such structures to have are: 
\begin{enumerate}
    \item they have to be able to capture all circuits of a given qubit size (i.e., surjectivity);
    \item their length has to be between one and two times the lower bound for surjectivity;
    \item they have to be compressible, \cmt{that is, they have to include a method for finding new structures with shorter length;}
    \item their maximum approximation error has to depend favorably on the compressed length;
    \item they \cmt{should} exactly compile special quantum gates (e.g., Toffoli gates) with length close to optimal;
    \item incorporating hardware constraints should only increase the length of the compressed structure by a constant multiplicative factor in order to preserve the same approximation error.
\end{enumerate}

Many of these properties are intertwined. For example, property (1) is necessary for property (4) and also for convergence (to a stationary point that only requires a global phase transformation), as discussed in the previous section. On the other hand, since the compressed structure will not be surjective when its length is less than the surjectivity bound, convergence becomes more difficult. Assuming property (2), then properties (1) and (5) are the same for all unitary matrices except a measure zero set. However, there are important matrices in the measure zero set, such as qubit permutations, controlled operations, and cyclic shifts \cite{beth2001quantum}, so we keep these requirements separate. But, as mentioned in Remark \ref{rmk:surj}, if we want to exactly compile these with minimal length, the mapping may be far from surjective and so convergence will be much more difficult. We will see this when compiling for Toffoli gates.

\cmt{We will consider three structures in this paper: $\cart$, standing for Cartan; $\sequ$, standing for sequential; and $\spin$, standing for spin.}

\begin{definition} Let $\cart\in J(\frac{23}{48}4^n-\frac{3}{2}2^n+\frac{4}{3})$ correspond to the QSD decomposition.
\end{definition}

\smallskip

As an example, the structure of CNOT units for $\cart(3)$ is
\smallskip

\begin{tikzpicture}
\node[scale=.65] {
\begin{quantikz}
&\ctrl{1} \cng &\ctrl{1}\cng&\ctrl{1}\cng&\qw&\ctrl{2}\cnq&\qw&\ctrl{1}\cng&\ctrl{1}\cng&\ctrl{1}\cng&\qw&\ctrl{2}\cnq&\qw&\ctrl{2}\cnq&\ctrl{1}\cng&\ctrl{1}\cng&\ctrl{1}\cng&\qw&\ctrl{2}\cnq&\qw&\ctrl{1}\cng&\ctrl{1}\cng&\ctrl{1}\cng&\qw\\
&\targ{}&\targ{}&\targ{}&\ctrl{1}\cng&\qw&\ctrl{1}\cng&\targ{}&\targ{}&\targ{}&\ctrl{1}\cng&\qw&\ctrl{1}\cng&\qw&\targ{}&\targ{}&\targ{}&\ctrl{1}\cng&\qw&\ctrl{1}\cng&\targ{}&\targ{}&\targ{}&\qw\\
&\qw&\qw&\qw&\targ{}&\targ{}&\targ{}&\qw&\qw&\qw&\targ{}&\targ{}&\targ{}&\targ{}&\qw&\qw&\qw&\targ{}&\targ{}&\targ{}&\qw&\qw&\qw&\qw&\rstick{.}
\end{quantikz}};
\end{tikzpicture}
\smallskip

where the CNOT boxes represent our CNOT unit~\cite{drury2008constructive}. $\cart$ satisfies properties (1) and (2) by the constructive proof of \cite{shende2006}. On the other hand, as is, it does not satisfy property (3) on compressibility. \cmt{So}, we will propose a method in Section \ref{sec:prune} that can be used to compress it to as short as the TLB with practically no error, thereby also supporting property (4). $\cart$ only increases by a multiplicative factor of 9 for the uncompressed structure, as shown in \cite{shende2006}. However, it is an open problem how to compile special gates close to their optimal length, and compress it in a way that does not increase connectivity, so the answer to properties (5) and (6) is unknown.

While $\cart$ has a recursive structure, the next two layouts we consider have repeating structures.

\begin{definition} Let $\sequ(L)\in J(L)$ correspond to $L$ CNOT units in order. There are $\frac{n(n-1)}{2}$ CNOT units and the order of $CNOT_{jk}$ is $n(j-1)+\left(k-\frac{j(j+1)}{2}\right)$. Once the end of the order is reached, it repeats.
\end{definition}

\smallskip

\begin{definition} Let $\spin(L)\in J(L)$ correspond to the structure that alternates between $(1\to 2,3\to 4,\ldots)$ and $(2\to 3,4\to 5,\ldots)$.
\end{definition}

\smallskip

As an example, the structure of CNOTs for $\sequ(12)$ when $n=3$ is
\smallskip

\begin{tikzpicture}
\node[scale=.65] {
\begin{quantikz}
&\ctrl{1}\cng&\ctrl{2}\cnq&\qw&\ctrl{1}\cng&\ctrl{2}\cnq&\qw&\ctrl{1}\cng&\ctrl{2}\cnq&\qw&\ctrl{1}\cng&\ctrl{2}\cnq&\qw&\qw\\
&\targ{}&\qw&\ctrl{1}\cng&\targ{}&\qw&\ctrl{1}\cng&\targ{}&\qw&\ctrl{1}\cng&\targ{}&\qw&\ctrl{1}\cng&\qw\\
&\qw&\targ{}&\targ{}&\qw&\targ{}&\targ{}&\qw&\targ{}&\targ{}&\qw&\targ{}&\targ{}&\qw
\end{quantikz}};
\end{tikzpicture}

\smallskip

and for $\spin(12)$ when $n=3$ is
\smallskip

\begin{tikzpicture}
\node[scale=.65] {
\begin{quantikz}
&\ctrl{1}\cng&\qw&\ctrl{1}\cng&\qw&\ctrl{1}\cng&\qw&\ctrl{1}\cng&\qw&\ctrl{1}\cng&\qw&\ctrl{1}\cng&\qw&\qw\\
&\targ{}&\ctrl{1}\cng&\targ{}&\ctrl{1}\cng&\targ{}&\ctrl{1}\cng&\targ{}&\ctrl{1}\cng&\targ{}&\ctrl{1}\cng&\targ{}&\ctrl{1}\cng&\qw\\
&\qw&\targ{}&\qw&\targ{}&\qw&\targ{}&\qw&\targ{}&\qw&\targ{}&\qw&\targ{}&\qw&\rstick{.}
\end{quantikz}};
\end{tikzpicture}

\smallskip

Both $\sequ(L)$ and $\spin(L)$ have the length of the circuit as an input, and so satisfy property (3). But for what $L$, if any, do they satisfy property (1) on surjectivity? We run experiments in Section \ref{sec:gd} to evaluate $\sequ(L)$ and $\spin(L)$ with respect to properties (1), (2), and (4). The results support the following conjecture about (1) and (2).

\begin{conjecture}
\label{conj:surj}
$V_{\sequ(L)}$ and $V_{\spin(L)}$ are surjective for $L\geq  \frac{1}{4}(4^n-3n-1)$, that is, the TLB.
\end{conjecture}

In Section \ref{sec:gd}, we also apply $\sequ(L)$ and $\spin(L)$ to specific quantum gate\cmt{s} (e.g., Toffoli), thereby supporting property (5), despite non-surjectivity. 

Structure $\sequ(L)$ can incorporate hardware constraints by simply skipping the CNOTs for unconnected qubits, but it is not obvious how much $L$ would have to increase to preserve a given maximum approximation error.

Fortunately, the results in the next section suggest $L$ would not have to increase at all. In the next section, we test $\sequ(L)$ with hardware constraints corresponding to both a star topology and a line topology. All three structures perform comparably to $\sequ(L)$ without any hardware constraints. Furthermore, $\spin(L)$, another structure that satisfies the line topology hardware constraints, also performs comparably. The experiments suggest that different repeating structures, as long as they include all of the qubits, may fill the space equally fast, on average.

In Table~\ref{tab:my_label_0}, we summarize our main findings for the particular structures we have discussed.

\begin{remark}
\label{rem:depth}
Quantum computers can implement \rev{arbitrary} gates on disjoint qubits simultaneously. \rev{Hence, they can implement CNOT units on disjoint qubits simultaneously.} \cmt{Thus, while our theory is in terms of length--the number of CNOTs in a structure, two other relevant metrics are CNOT depth---the minimum number of layers of CNOT units in a structure---and circuit depth---the minimum number of layers in a structure. Note that if all rotation angles are non-zero, then circuit depth is three plus three times CNOT depth. Also, CNOT depth is always smaller than or equal to the length.} In particular, $\spin(L)$ can be implemented with \cmt{CNOT depth} $\lceil 2L/(n-1)\rceil$ for $n>3$. As another example, $\spin(12)$ when $n=4$ is
\smallskip

\begin{tikzpicture}
\node[scale=.65] {
\begin{quantikz}
&\ctrl{1}\cng&\qw&\qw&\ctrl{1}\cng&\qw&\qw&\ctrl{1}\cng&\qw&\qw&\ctrl{1}\cng&\qw&\qw&\qw\\
&\targ{}&\qw&\ctrl{1}\cng&\targ{}&\qw&\ctrl{1}\cng&\targ{}&\qw&\ctrl{1}\cng&\targ{}&\qw&\ctrl{1}\cng&\qw\\
&\qw&\ctrl{1}\cng&\targ{}&\qw&\ctrl{1}\cng&\targ{}&\qw&\ctrl{1}\cng&\targ{}&\qw&\ctrl{1}\cng&\targ{}&\qw\\
&\qw&\targ{}&\qw&\qw&\targ{}&\qw&\qw&\targ{}&\qw&\qw&\targ{}&\qw&\qw
\end{quantikz}};
\end{tikzpicture}

\smallskip
\cmt{which becomes, if we implement CNOT units on disjoint pairs of qubits simultaneously,}
\smallskip

\begin{tikzpicture}
\node[scale=.65] {
\begin{quantikz}
&\ctrl{1}\cng&\qw&\ctrl{1}\cng&\qw&\ctrl{1}\cng&\qw&\ctrl{1}\cng&\qw&\qw\\
&\targ{}&\ctrl{1}\cng&\targ{}&\ctrl{1}\cng&\targ{}&\ctrl{1}\cng&\targ{}&\ctrl{1}\cng&\qw\\
&\ctrl{1}\cng&\targ{}&\ctrl{1}\cng&\targ{}&\ctrl{1}\cng&\targ{}&\ctrl{1}\cng&\targ{}&\qw\\
&\targ{}&\qw&\targ{}&\qw&\targ{}&\qw&\targ{}&\qw&\qw&\rstick{.}
\end{quantikz}};
\end{tikzpicture}

\smallskip

\cmt{On the other hand, imposing the line topology hardware constraints on $\sequ(12)$ results in}
\smallskip

\begin{tikzpicture}
\node[scale=.65] {
\begin{quantikz}
&\ctrl{1}\cng&\qw&\qw&\ctrl{1}\cng&\qw&\qw&\ctrl{1}\cng&\qw&\qw&\ctrl{1}\cng&\qw&\qw&\qw\\
&\targ{}&\ctrl{1}\cng&\qw&\targ{}&\ctrl{1}\cng&\qw&\targ{}&\ctrl{1}\cng&\qw&\targ{}&\ctrl{1}\cng&\qw&\qw\\
&\qw&\targ{}&\ctrl{1}\cng&\qw&\targ{}&\ctrl{1}\cng&\qw&\targ{}&\ctrl{1}\cng&\qw&\targ{}&\ctrl{1}\cng&\qw\\
&\qw&\qw&\targ{}&\qw&\qw&\targ{}&\qw&\qw&\targ{}&\qw&\qw&\targ{}&\qw
\end{quantikz}};
\end{tikzpicture}

\smallskip

\cmt{which becomes, if we implement CNOT units on disjoint pairs of qubits simultaneously,}
\smallskip

\begin{tikzpicture}
\node[scale=.65] {
\begin{quantikz}
&\ctrl{1}\cng&\qw&\ctrl{1}\cng&\qw&\ctrl{1}\cng&\qw&\ctrl{1}\cng&\qw&\qw&\qw\\
&\targ{}&\ctrl{1}\cng&\targ{}&\ctrl{1}\cng&\targ{}&\ctrl{1}\cng&\targ{}&\ctrl{1}\cng&\qw&\qw\\
&\qw&\targ{}&\ctrl{1}\cng&\targ{}&\ctrl{1}\cng&\targ{}&\ctrl{1}\cng&\targ{}&\ctrl{1}\cng&\qw\\
&\qw&\qw&\targ{}&\qw&\targ{}&\qw&\targ{}&\qw&\targ{}&\qw&\rstick{.}
\end{quantikz}};
\end{tikzpicture}

\smallskip

\cmt{Structurally, $\spin(12)$ and $\sequ(12)$ are almost identical on the line topology, except that $\spin(12)$ moves the last CNOT unit of $\spin(12)$ to the first layer and so decreases the CNOT depth from 9 to 8.}
\end{remark}

\begin{table}
    \centering
        \caption{Considered structures and their properties. Numbered references are to subsections.}
    \begin{tabular}{ccccccc}
    \toprule 
    Structure & \multicolumn{6}{c}{Properties} \\
    & (1) & (2) & (3) & (4) & (5) & (6) \\
    \midrule 
    $\cart$ & \cite{drury2008constructive} & \cite{drury2008constructive} & \ref{subsec:prune}&\ref{subsec:prune} & ? & ?  \\
    $\sequ(L)$ & \multicolumn{2}{c}{Conjecture~\ref{conj:surj}} & \checkmark & \ref{subsec:random} & \ref{subsec:toffoli} & \ref{subsec:random}  \\
    $\spin(L)$ & \multicolumn{2}{c}{Conjecture~\ref{conj:surj}} & \checkmark & \ref{subsec:random} & \ref{subsec:toffoli} & \checkmark  \\
    \bottomrule
    \end{tabular}
    \label{tab:my_label_0}
\end{table}

\section{Gradient descent}
\label{sec:gd}

In Section~\ref{sec:math-opt}, we formulated the mathematical optimization problem, and in Section~\ref{sec:propr} we discussed some of its properties once specified to CNOT unit structures. In the previous section, we discussed some special layouts. We are now ready to look at the approximate compiling problem for a fixed structure. Specifically, we \cmt{look} at the problem
\begin{equation}\label{eq:aqcp-fixstr}
\textbf{(AQCP-$\btheta$)}\qquad \min_{\btheta\in [0, 2\pi)^{p}, \ct = \overline{\ct}} \, f_{\ct}(\btheta):=\frac{1}{2}\|\Vct(\btheta) - U\|_{\mathrm{F}}^2,
\end{equation}
where $\ct$ is fixed to one of the mentioned structures ($\cart, \sequ(L), \spin(L)$), which in turn specifies how many rotation angles there are (i.e., $p$), and we will use a first-order method to find second-order stationary points. The specific method can be tuned in practice, so we present gradient descent and then discuss the modifications that can be made for its variants.

First, we randomly initialize $\btheta[0]$ from the uniform distribution on $[0, 2\pi)^{p}$. Then, we compute
\begin{equation}\label{eq:gradient}
\btheta[t+1] = \btheta[t] - \alpha \nabla f_{\ct}(\btheta[t])
\end{equation}
until the stopping criteria, $\|\nabla f_{\ct}(\btheta[t])\|\leq \epsilon$, is met. The step-size, $\alpha$, is tuned in practice, and the final $\btheta$ is then wrapped in the $[0, 2\pi)^{p}$ set. Component-wise, Eq.~\eqref{eq:gradient} reads,
\begin{equation}\label{eq:gradient-c}
\theta_k[t+1] = \theta_k[t] + \alpha \Re~\tr\bigg[\partialk \Vct(\btheta)\dagg U\bigg].
\end{equation}

One variant of gradient descent is Nesterov's method~\cite{nesterov1983}, which applies the gradient descent step to an auxiliary sequence. The auxiliary sequence is essentially the main sequence but with so-called ``momentum.'' For smooth convex optimization, Nesterov's method is optimal among first-order methods.

Another modification that can be made to gradient descent is the injection of noise. On the one hand, diminishing Gaussian noise can be added to the gradient, in which case the method is called gradient descent with Langevin dynamics. On the other hand, \cmt{noise uniformly sampled from a fixed radius ball} can be added in a random direction whenever a ``stuck'' criteria is met, in which case the method is called perturbed gradient descent.

We found that it was not necessary to add noise and that Nesterov's method was much faster than gradient descent, so we used it as our method of choice. \cmt{As mentioned previously, random initialization helps escape saddle points to find second-order stationary points. However, the random prior may also affect what kind of second-order stationary points are found. In the non-surjective setting, we would like to avoid spurious local minima as well as saddle points, so it is a future research direction to consider different initializations.}

\subsection{A note on computational complexity}

Note that the gradient computation makes up the majority of the computation and memory complexity. Specifically, we have to compute $\partialk \Vct(\btheta)$ for each $k\in\{1, \dots,p\}$, each of which takes the same number of flops to compute as $\Vct(\btheta)$. To compute the matrix for a single layer, via the kronecker product, takes $O(d^2)$ flops. To multiply two matrices takes $O(d^3)$ flops. There are $L$ layers so the total computational complexity for $\Vct(\btheta)$ is $O(Ld^3)$. Thus, computing $\partialk \Vct(\btheta)$ from scratch for each $k\in\{1, \dots,p\}$ comes out to $O(L^2d^3 )$. The memory complexity is $O(L+d^2)$. Since this is the dominant part of the algorithm, the total computational complexity is $O(L^2d^3 T)$ (where $T$ is the number of iterations) and the memory complexity is $O(L+d^2)$.

Fortunately, it is possible to trade-off between the computational and memory complexity. \cmt{We do this in a similar way to backpropogation, the algorithm for computing the gradient of the loss of a feedforward neural network applied to a sample point \cite{goodfellow2016deep}, which is an example of reverse mode automatic differentiation \cite{linnainmaa1976taylor}.} Instead of computing $\partialk \Vct(\btheta)$ from scratch for each $k\in\{1, \dots,p\}$, we can store each of the intermediate matrix computations. Specifically, we can compute the matrix for each layer, taking $O(Ld^2)$ flops and storage. Then we can compute the matrix multiplications from both the right and left, storing the intermediate outputs. That takes $O(Ld^3)$ flops and $O(Ld^2)$ storage. Then, we compute four new versions of each layer corresponding to the partial derivative of each of the four parameters in a layer. This takes $O(Ld^2)$ flops and storage again. Finally, for the partial derivative of each parameter, we make only two matrix computations, taking $O(Ld^3)$ flops and $O(d^2)$ storage (since we don't store all $p$ matrices at once, we multiply them by $U$ and compute the real part of the trace to get a real number). Thus, this way of computing the gradient descent iterates amounts to $O(Ld^3T)$ flops and $O(Ld^2)$ storage. The reduction in computational complexity from $L^2$ to $L$ is significant since $L$ \cmt{can be} exponential in the number of qubits. \cmt{In particular, when the goal is exactly compiling general unitary matrices, we need $L=\Omega(4^n)=\Omega(d^2)$. On the other hand, if the goal is exactly compiling specific unitary matrices or approximately compiling general unitary matrices, $L$ may be smaller. For example, to exactly compile the Toffoli gate, we only need $L=\Omega(n)=\Omega(\log_2(d))$.}

These analyses of the computation and memory complexity of gradient descent are for its implementation on a classical computer. However, it can be implemented on a quantum computer as well. Algorithm 3 of \cite{khatri2019} explains how to implement it using the power-of-two-qubits circuit. While this is a possibility, we do not follow it here. 

\subsection{Numerical tests: random unitary matrices, towards [G1]}
\label{subsec:random}

We start by testing the structures and gradient descent on random unitary matrices, which is our goal [G1]. We consider three different hardware connectivity graphs with edge sets
\begin{align*}
    E_1 &= \{(j,k)\mid j,k\in \{1, \ldots, n\},j\neq k\} \qquad & \textrm{(Full connectivity)}\\
    E_2 &= \{(1,j+1)\mid j\in\{1, \ldots, n-1\}\}\qquad & \textrm{(Star connectivity)}\\
    E_3 &= \{(j,j+1)\mid j\in\{1, \ldots, n-1\}\}\qquad & \textrm{(Line connectivity)}.
\end{align*}
$E_1$ corresponds to no hardware constraints, $E_2$ corresponds to the star topology, and $E_3$ corresponds to the line topology. Corresponding to these, we consider the structures $\sequ_{E_1}$, $\sequ_{E_2}$, $\sequ_{E_3}$, and $\spin$. Note that $\spin$ corresponds to $E_3$, but $\sequ_{E_3}\neq\spin$ except when $n\leq 3$.

\paragraph{Full connectivity results. }
We ran the experiments for $n=3$ and $n=5$ qubits. In both cases, we sampled $100$ random unitary matrices for each structure for different values of $L$. We ran gradient descent for each sample and then computed the approximation error and its probabilities. These curves are given in Figure~\ref{fig.fullconn} and Figure \ref{fig.fullconn2}. As one can see, both $\sequ$ and $\spin$ are almost identical on how the error depends on the length $L$, thereby suggesting that the actual structure does not \cmt{make} a big difference. In addition, both structures reach zero error around their theoretical lower bound (which can be computed as $L=14$ for $n=3$, and $L=252$ for $n = 5$), supporting our Conjecture~\ref{conj:surj}. \cmt{Figure \ref{fig.fullconn2} gives probabilistic statements: the $x$-axis represents a desired maximum approximation error (bound), while the $y$-axis reports the probability that starting from a random guess for a random unitary, we obtain an error larger than the desired bound. Note that this corresponds to one minus the CDF of the approximation error. Each $L$ considered has a curve in Figure~\ref{fig.fullconn2}, representing the distribution of the approximation error, while in Figure~\ref{fig.fullconn} it is compressed to a single point by taking either the mean or the maximum.}

\begin{figure}
\centering
\includegraphics[width=0.475\textwidth]{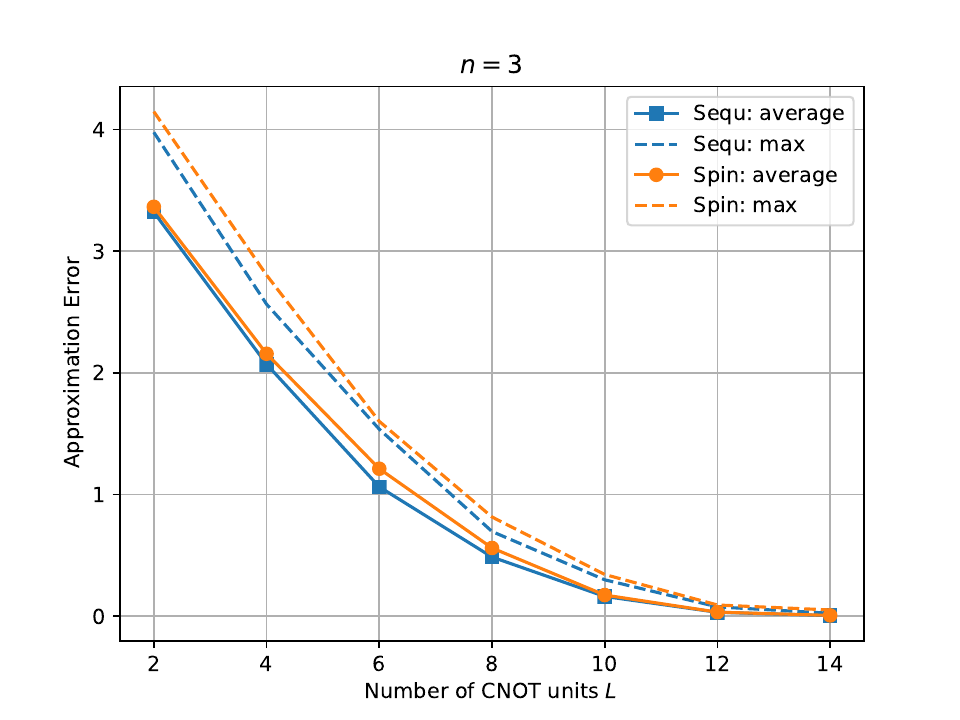}\includegraphics[width=0.475\textwidth]{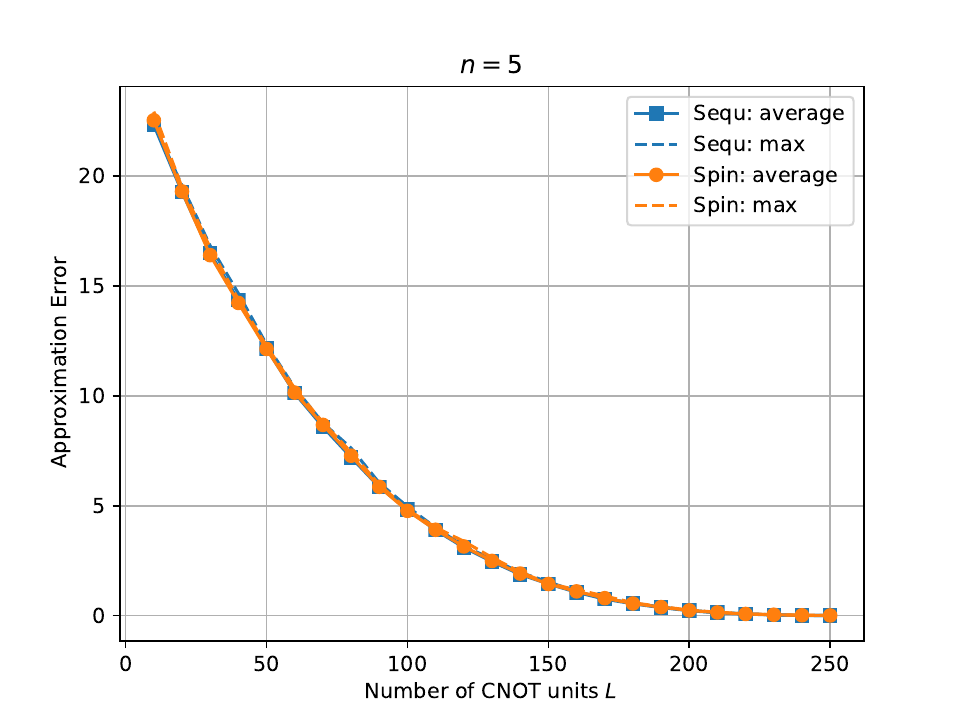}
\caption{Approximation error, \rev{$f_{\ct}$}, for $n=3$ (left) and $n=5$ (right).}
\label{fig.fullconn}
\end{figure}

\begin{figure}
\centering
\includegraphics[width=0.475\textwidth]{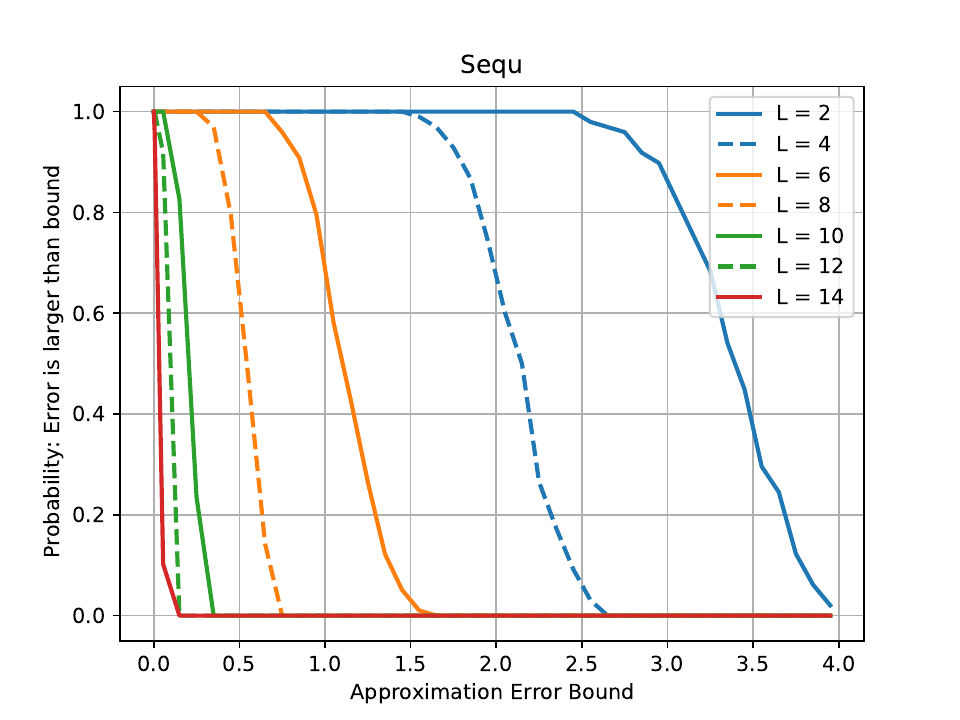}\includegraphics[width=0.475\textwidth]{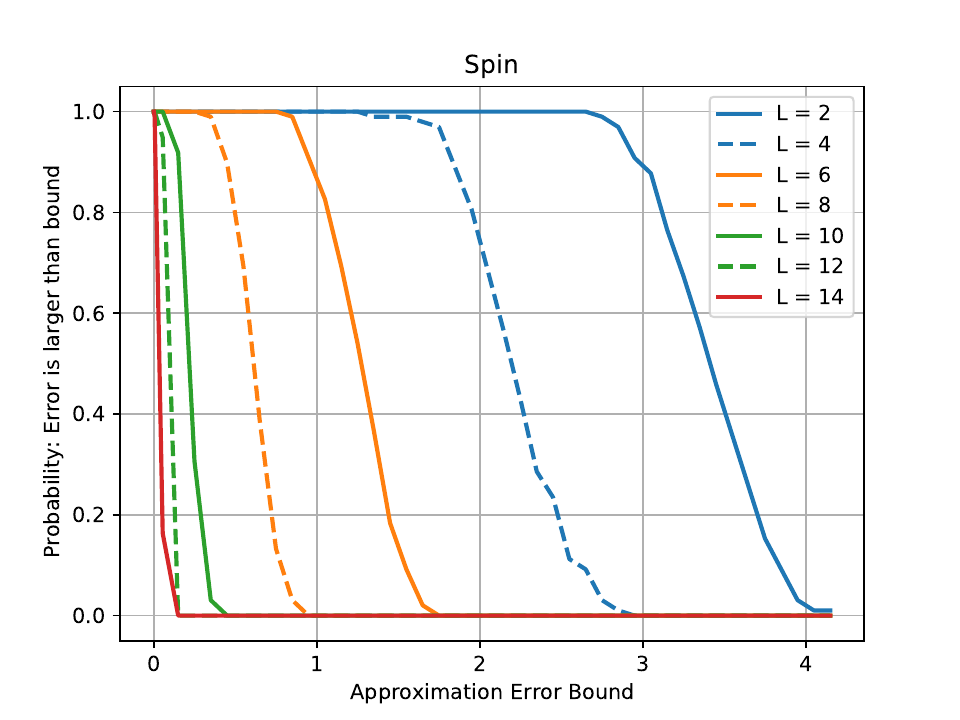}
\caption{Error probabilities for $\sequ$ and $\spin$ on $3$-qubits for different values of $L$ \cmt{that starting from a random initial point for a random unitary we obtain an approximation error larger than a desired bound. }}
\label{fig.fullconn2}
\end{figure}

\paragraph{Limited connectivity results. }

With the same settings, we investigate the effect of different connectivity patters on $\sequ$, for $n=5$ qubits. 
For $n=5$, the theoretical lower bound is 252 CNOTs. Figure \ref{fig.sparseconn} shows that $\sequ$ on all three connectivity patterns \cmt{approach} zero approximation error around its theoretical lower bound \cmt{(in the worst case, the max lines arrive at about $0.7$ error, which corresponds to a fidelity of $96\%$)}. Furthermore, the dependence of the approximation error on the length is a convex curve that is basically identical for all three patterns.

From an engineering perspective, it is quite surprising that the structures with limited hardware connectivity perform as well as \cmt{$\sequ$}, corresponding to full connectivity. However, from the ``separating parameters'' perspective that was used to derive the lower bound, this is less surprising: $\sequ$ seem to ``fill out'' the quantum circuit \cmt{and} separate the parameters equally \cmt{as} well on a line \cmt{as} on a fully connected graph.

This is in sharp contrast with the QSD which uses 20 and 444 CNOT units for $n=3$ and $n=5$, respectively, and increases in length by a multiplicative factor of 9 for the line topology. 

These experiments suggest that both $\spin$ and $\sequ$ are surjective for $L$ equal to the lower bound and that their approximation error depends favorably on the length, supporting properties (1) through (4). Regarding property (6), the experiments suggest that the length does not have to increase at all (\cmt{or perhaps very slightly, in the worst case scenario}).

\begin{remark}
\cmt{It is important to note here that the fact that the connectivity does not matter in a random unitary setting is a general statement about the overall smooth optimization landscape. We are saying that, if you take a random unitary, and a few random initial points, on average, the approximation error you obtain depends only on the number of CNOT units. This statement says that even if the different structures have different local minima and properties, with a macro-scale lens (i.e., in the worst-case), they behave very similarly. On the other hand, when one is interested in compiling a particular gate (say a Toffoli), in the least number of CNOT units, the structure plays an important role as discussed in the literature, e.g., \cite{Rev3-1,Rev3-2}, and as we will see next. In this case, the fact that connectivity \cmt{matters} is a statement about best case scenarios.}
\end{remark}

\begin{figure}
\centering
\includegraphics[width=0.475\textwidth]{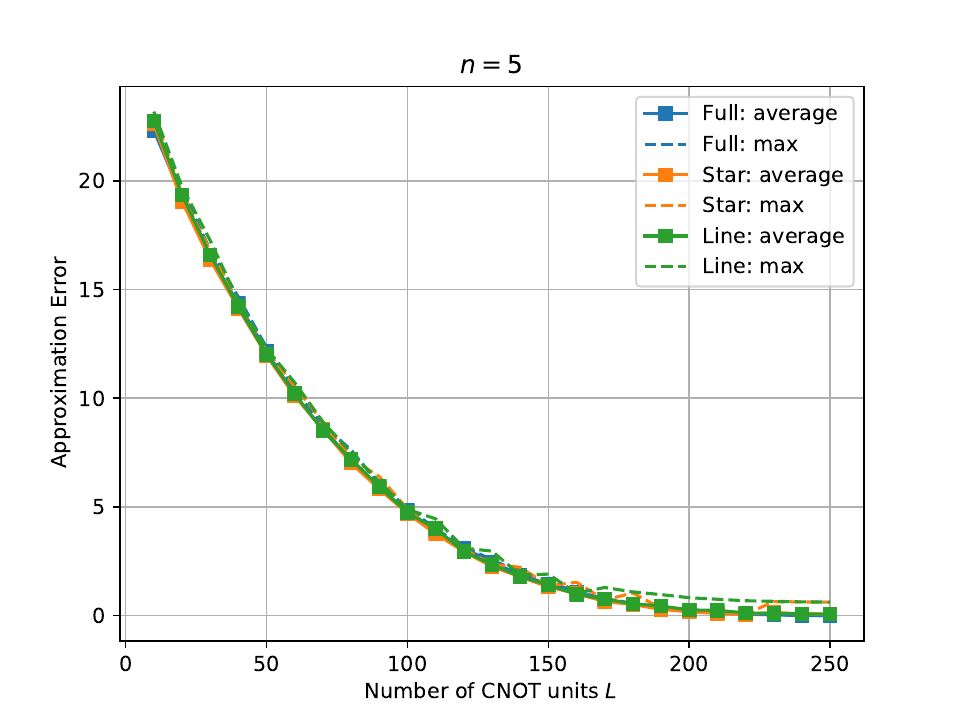}\includegraphics[width=0.275\textwidth, trim=-1cm -2cm 0 0 , clip = on]{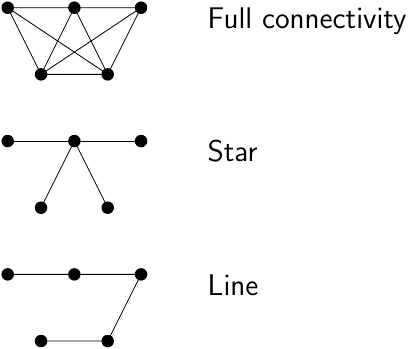}
\caption{Approximation error $n=5$ with different connectivity patterns.}
\label{fig.sparseconn}
\end{figure}

\subsection{Numerical tests: Special gates, towards [G2]}
\label{subsec:toffoli}

The only property we are left to explore in this section is property (5), meaning how close to optimal are the various structures in the case of special (often used and well-studied) gates. So, in addition to compiling random unitary matrices as efficiently as possible, we want to recover the shorter lengths that some important gates allow. For example, $k$-controlled $(n-k)$-qubit operations can be compiled with $O(n)$ gates if $k<n-1$ and $O(n^2)$ gates if $k=n-1$ \cite[Ex. 4.2]{beth2001quantum}\cite[Lems. 7.2 and 7.5]{barenco1995}. In particular, the $n$-qubit Toffoli gate, also known as the multi-controlled-X gate, can be thought of as a $k$-controlled $(n-k)$-Toffoli gate, and so can be compiled with $O(n)$ gates. Furthermore, \cite{shende-lowerbound} gives a lower bound of $2n$ CNOTs for the $n$-qubit Toffoli gate (other lower bounds are possible when considering circuits with ancillae~\cite{dmit}). The lower bound is tight for $n=3$. Qiskit~\cite{Qiskit2019} decomposes the 3-qubit Toffoli gate into 6 CNOTs and the 4-qubit Toffoli gate into 14 CNOTs\footnote{A 14 CNOT implementation can be obtained from the 20 CNOT one of~\cite{barenco1995} by substituting the controlled-Vs with their 2-CNOT implementations and applying the templates presented in~\cite{redu2008}.}. We remark that these numbers of CNOTs are significantly lower than the \cmt{TLB}, and therefore property (5) is not trivial to fulfill. 

In addition to the Toffoli gate, we also \cmt{consider} the Fredkin 3 qubit gate, as well as the 1-bit full adder (which is a 4 qubit gate). The list of important gates is by no means exhaustive, but it already gives a glimpse of how well $\sequ$ and $\spin$ work for important gates.



To test property (5), we ran our gradient descent algorithm on $\spin$ and $\sequ$ many times (we report in Table~\ref{tab:my_label} the results, as well as the number of exact compilations divided by the number of tries). We compare with the Qiskit compilations on full connectivity and on line connectivity \rev{(with the usual workflow, one would compile e.g., a Toffoli gate on the full connectivity hardware and then add swap gates to transpile it on the limited connectivity one)} \cite{Qiskit2019}.

In the case of $\sequ$, we also consider the case of permuting the CNOT units to span more possibilities: this is, at the moment, a random search permuting all the possibilities, which is unpractical for large $L$'s (for instance for $L=14$ it would amount to over \rev{$150$ million} possibilities), but gives us a glimpse that structure does matter when compiling very specific gates in the non-surjective domain. 

The results indicate that, e.g., $\sequ(18)$ and $\spin(18)$ can exactly compile a 4-qubit Toffoli, which is better than \cmt{the} 20 CNOT implementation \cmt{of~\cite{barenco1995}}. And allowing for permutations in the sequence of CNOTs, even $\sequ(14)$ can do it, which is optimal\footnote{The ease at which we have found this, despite the \rev{$14!/3!/3!/2!/2!/2!/2!$} possible structures, indicates that there may be more than one structure that can deliver an optimal compilation. \rev{Note that, in general, for $m$ types of CNOTs with corresponding quantities $L_1,\ldots,L_m$, there are $(L_1+\cdots+L_m)!/(L_1!\cdots L_m!)$ permutations.}}. The solution we have found is reported in Figure~\ref{fig.4qt}, and it uses a different sequence of CNOT gates than the one in Qiskit. 

What we observe is that there are different possibilities in compiling a certain circuit exactly, and in most cases the optimal solutions \cmt{are} different. This shows that our algorithm can be used as a tool to discover new exact compilations of special gates. 

The experiment suggests that property (5) is satisfied to a reasonable extent for both $\sequ$ and $\spin$.

\cmt{Finally, by looking at the performance of Qiskit when compiling on a line connectivity and comparing it to our $\spin$ structure, which enforces the line connectivity by design, we further appreciate the advantage of our method in terms of CNOT count.  
}

\begin{table}[]
    \centering
    \caption{Exact compilations of special gates. In \cmt{parentheses} the number of successful compilations vs. the number of trials starting with a different initial condition, and in the case of $\sequ$ with permutation, with a different permutation of the CNOT layout. \cmt{Note that the Qiskit compilation (opt level of 3 and ``sabre'' as routing and layout) is stochastic and can return a variable number of CNOTs.} $^*$After several trials, we have found a satisfactory layout, and the numbers correspond to this one.}
    \label{tab:my_label}
    \begin{tabular}{ccccccc}
    \toprule 
    Gate & \multicolumn{2}{c}{Qiskit CNOTs}  & \multicolumn{2}{c}{$\sequ$} & $\spin$ \\
    & Full conn. & Line conn. &  w/o perm. & w perm & \\ \toprule
    Toffoli 3 qubit & 6 & \cmt{7-9} & 7 (19/100) & 6 (27/100) & 8 (38/100) \\
    Toffoli 4 qubit & 14 & \cmt{36}  & 18 (4/350) & 14 (1/100)* & 18 (1/350)  \\
    Fredkin 3 qubit & 7 & \cmt{10} & 8 (55/100)  & 7 (4/100) & 8 (31/100) \\
    1-bit full adder 4 qubit & 10 & \cmt{16} & 10 (11/100) & 10 (3/100) & 14 (8/500)  \\ \bottomrule
    \end{tabular}
\end{table}

\begin{figure}[h]
    \centering
    \includegraphics[width=1.\linewidth]{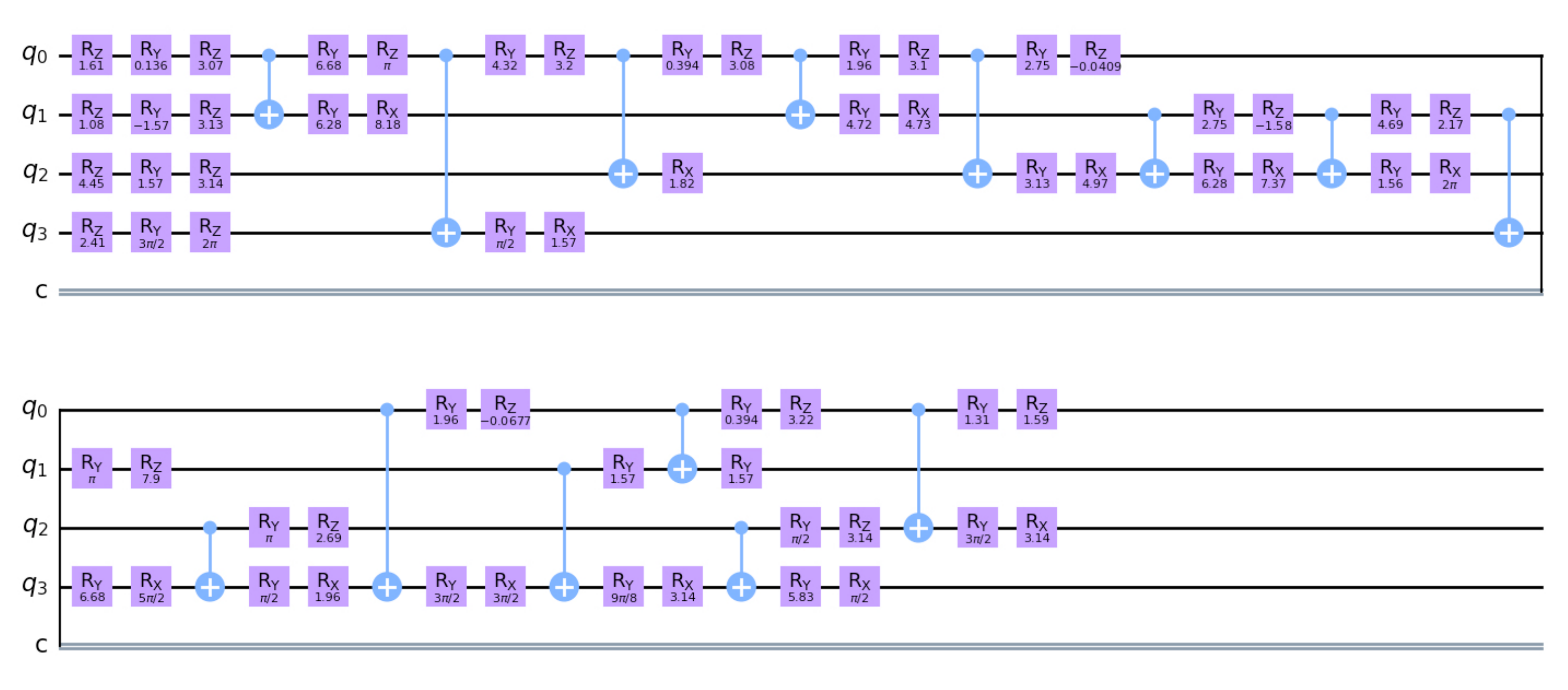}
    \caption{A different from Qiskit 4-qubit Toffoli exact compilation, obtained by permuting the CNOT units in the $\sequ(14)$ structure and running gradient descent with many initial points. }
    \label{fig.4qt}
\end{figure}


\section{Circuit compression via regularization}
\label{sec:prune}

We move now to address property (3) on compressibility in more detail. In particular, motivated by the $\cart$ structure, we notice that, while $\cart$ is provably surjective with length twice the lower bound, thus satisfying properties (1) and (2), it is not clear how to compress it and so address property (3). The purpose of this section is to develop an algorithmic technique for compressing arbitrary structures, and $\cart$ in particular. Rather than choosing which CNOTs to keep in a structure (equivalently, which CNOTs to eliminate) \cmt{\textit{a priori}}, we consider the question \cmt{of} how to design an algorithm that automatically finds the best compression for the target unitary $U$.

We approach this problems with two ideas: (1) we know that there are techniques to reduce consecutive CNOTs when no rotation gates are between them (let us call \cmt{these} compaction rules); (2) we know that, when optimizing for the angles, we can enforce sparsity of the solution by adding a pertinent regularization. 

A closer look at (2) inspires us to enforce groups of four rotation angles following the CNOTs to be zero, thereby eliminating all the rotation gates after a CNOT. This leads naturally to a group LASSO regularization, and we explore it in Section~\ref{subsec:lasso}. 

A closer look at (1) suggests special compaction rules, which we discuss in Section~\ref{subsec:group}. 

Finally, the complete algorithm starts by enforcing sparsity, eliminating zero rotation gates, compressing the structure via compaction rules, and then re-optimizing with the regular gradient descent on the compacted structure. This algorithm is discussed in Section~\ref{subsec:prune}.

\subsection{The ``synthesis'' algorithm}
\label{subsec:group}

The guiding question for this subsection is: given a list of CNOTs, can we find a shorter list of CNOTs such that the matrix product, in $\SU(2^n)$, of the first list and the second list are equal. While the research works in this area are many~\cite{patel2008optimal,aaronson2004improved,amy2017finite,maslov2018shorter,garion2020structure,bravyi2020hadamard,beth2001quantum}, we use here an adapted version of the ``synthesis'' algorithm of~\cite{patel2008optimal}. There the authors \cmt{give} an asymptotically optimal synthesis of CNOT circuits. The idea is that CNOTs on $n$-qubits can be identified with elementary matrices in $\GL(n,\Z_2)$. 
Given $A\in\GL(n,\Z_2)$, we can compute its LU decomposition in terms of elementary matrices. On the other hand, things are easier in our case, since we \cmt{only consider} downward-facing CNOTs, which can be related to $\LT(n,\Z_2)$, the group of $n$ by $n$ lower unitriangular matrices on $\Z_2$. This leads us to implement an adapted ``synthesis'' algorithm that consists of three steps: identifying CNOTs with their corresponding matrices in $\LT(n,\Z_2)$, multiplying them in $\LT(n,\Z_2)$, and reading off the locations of 1's in the product. We report the following theorem concerning this algorithm.

\begin{theorem}{\cite{patel2008optimal}}
Given an $L$-long circuit of downward-facing CNOTs, the ``synthesis'' algorithm correctly outputs an equivalent circuit of downward-facing CNOTs of length $\leq n(n-1)/2$ in $O(n^3 L)$ computations.
\end{theorem}

The correctness of the algorithm follows from \cite{patel2008optimal}, the length follows from the fact that there are $n(n-1)/2$ lower-triangular entries in an $n$ by $n$ matrix, and the run-time is based on $L-1$ matrix multiplications. While the algorithm may output a word of length $n(n-1)/2$, it is possible for words to be further reduced. Hence, this ``synthesis'' algorithm is not optimal, but its simplicity is appealing. The lower bound on word lengths is $\Omega(n^2/\log(n))$ for $\GL(n,\Z_n)$, and one way to obtain it is to partition the matrix into blocks, as is done in \cite{patel2008optimal}. \cmt{Note that words can also be reduced via} identities on three qubits, namely commutation rules, cancellations of two subsequent identical CNOTs, and mirror rules (see, e.g, ~\cite{rules}). \cmt{However, these have little effect for $n>3$.}

We remark that, in general, the ``synthesis'' algorithm does not respect hardware connectivity, so it will be used only in cases in which hardware connectivity is not an issue. \cmt{However, the three qubit identities do} maintain the hardware connectivity constraints of the original circuit, so they can be used in all the cases. \cmt{Note that there are recent works using Steiner trees to apply the synthesis algorithm in a way that does respect hardware connectivity: \cite{nash2020quantum,wu2019optimization,kissinger2020cnot,de2020architecture,gheorghiu2020reducing}.}

\subsection{Setting parameters to zero}
\label{subsec:lasso}

The guiding question for this subsection is: which parameters should we set to zero to get a good compressed structure via the techniques presented in the previous subsection\cmt{?} As mentioned, our approach is to enforce sparsity \cmt{of} the resulting $\btheta$ vector by pushing groups of the four-rotation angles following a CNOT to be zero (i.e., all the rotations of a given CNOT unit). This can be achieved via a group Lasso regularization~\cite{GLasso2013} (see also \cite{becker2011nesta,beck2009fast,candes2008enhancing}), as follows. The collection of rotation angles for each of the CNOT units is the vector $\btheta_{\ell} := [\theta_{3n+4\ell-3}, \ldots, \theta_{3n+4\ell}]$, with $\ell = \{1, \ldots, L\}$. Enforcing each of these vectors to be zero, amounts to adding a regularization of the form $\|\btheta_{\ell}\|_2$ for each group, and therefore solving the problem:
\begin{equation}\label{eq:aqcpg-prune}
\textbf{(AQCP-$\btheta,\lambda$)}\qquad \min_{\btheta\in [0, 2\pi)^{p}, \ct = \overline{\ct}} \, f_{\ct}(\btheta;\lambda):=\frac{1}{2}\|\Vct(\btheta) - U\|_{\mathrm{F}}^2+\lambda\sum_{\ell=1}^L \|\btheta_{\ell}\|_2,
\end{equation}
with regularization parameter $\lambda>0$, which trades off approximation error and \cmt{sparsity}. 

We can solve~\eqref{eq:aqcpg-prune} by a proximal gradient descent, as done in~\cite{GLasso2013}; although problem~\eqref{eq:aqcpg-prune} is non-convex, we have similar convergence results as \cmt{for} gradient descent, meaning that for small regularization parameters $\lambda$, we can show convergence of perturbed proximal gradient descent to a second-order stationary point~\cite{huang2019perturbed}.

In particular, proximal gradient descent amounts to computing a gradient descent \cmt{step} and then applying a proximal operator (which, in this case, is the block-wise soft-thresholding operator \cite[Eq.(7)]{mosci2010solving}). Starting from a randomly initialized $\btheta[0]$, this yields the component-wise recursion for $k=3n+4\ell-3,\ldots,3n+4\ell$, for $\ell=1,\ldots,L$, as
\begin{subequations}
\label{eq:gproxgradient-c}
\begin{eqnarray}
[\btheta_{\ell}^+]_{k} &=& \theta_k[t] + \alpha \Re~\tr\bigg[\frac{\partial}{\partial\theta_k} \Vct(\btheta[t])\dagg U\bigg] \quad \cmt{\text{for all }} k\in\{3n+4\ell-3,\ldots,3n+4\ell\} \\
\theta_k[t+1] &= &\frac{[\btheta_{\ell}^+]_{k}}{\|\btheta_{\ell}^+\|}(\|\btheta_{\ell}^+\|-\alpha\lambda)_+ \qquad \cmt{\text{for all }} k\in\{3n+4\ell-3,\ldots,3n+4\ell\}\\
\theta_k[t+1] &=& \theta_k[t] + \alpha \Re~\tr\bigg[\frac{\partial}{\partial\theta_k} \Vct(\btheta[t])\dagg U\bigg] \qquad \cmt{\text{for all }} k\in\{1, \ldots, 3n\},
\end{eqnarray}
\end{subequations}
where $(y)_+=\max\{y,0\}$.

Eq.s~\ref{eq:gproxgradient-c} are block-wise recursions: for each CNOT unit, we run the gradient descent for each angle of the unit, and then run the proximal operator. For the angles not belonging to the CNOT units, then it is business as usual.

\subsection{Compression algorithm}
\label{subsec:prune}

Integrating the ideas from the previous two subsections, we propose the following algorithm for compressing an arbitrary structure.

\begin{algorithm}[H]
\begin{algorithmic}[1]
\Require $U,\ct,\lambda$, initial condition $\btheta[0]$
\State compute $\btheta^*_{\textrm{GL}}$ via proximal gradient descent~\eqref{eq:gproxgradient-c} on $f_{\ct}(\cdot,\lambda)$
\State Eliminate all the rotation gates that have zero rotation angle
\State compress $ct$ via the ``synthesis'' algorithm
\State further compress $ct$ via CNOT identities
\State compute $\btheta^*$ via standard gradient descent~\eqref{eq:gradient} on $f_{\ct'}$, where $\ct'$ is the compressed structure
\State \Return compressed structure, $\ct'$, and corresponding angles, $\btheta^*$
\end{algorithmic}
\caption{Compression via group LASSO}
\label{algo:LASSO}
\end{algorithm}

The algorithm consists \cmt{of} running proximal gradient descent~\eqref{eq:gproxgradient-c}, starting with some given initial condition $\btheta[0]$, on the regularized cost $f_{\ct}(\cdot,\lambda)$. Then, we eliminate all the rotation gates that have zero rotation angle and apply the compaction rules to reduce the circuit. Finally, we run standard gradient descent~\eqref{eq:gradient} on $f_{\ct'}$, where $\ct'$ is the compressed structure, to compute the best parameters $\btheta^*$ for the compressed structure. 

Line 3 (i.e., ``synthesis'') can be included, when hardware connectivity constraints are not important, or not, when they are. 

\begin{figure}
    \centering
    \includegraphics[width=.475\linewidth]{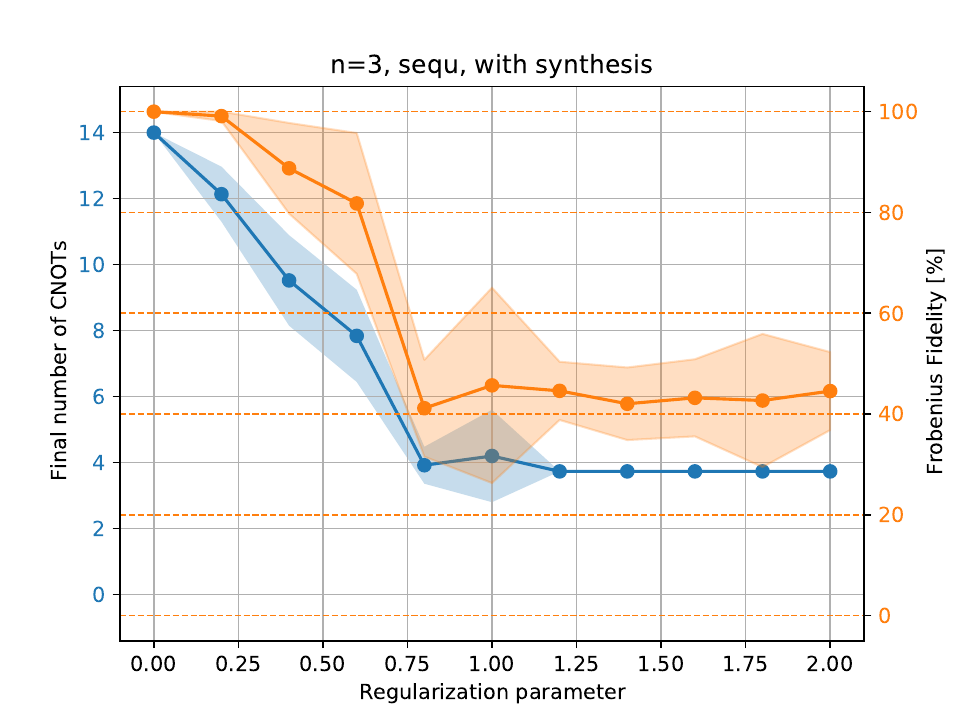}
    \includegraphics[width=.475\linewidth]{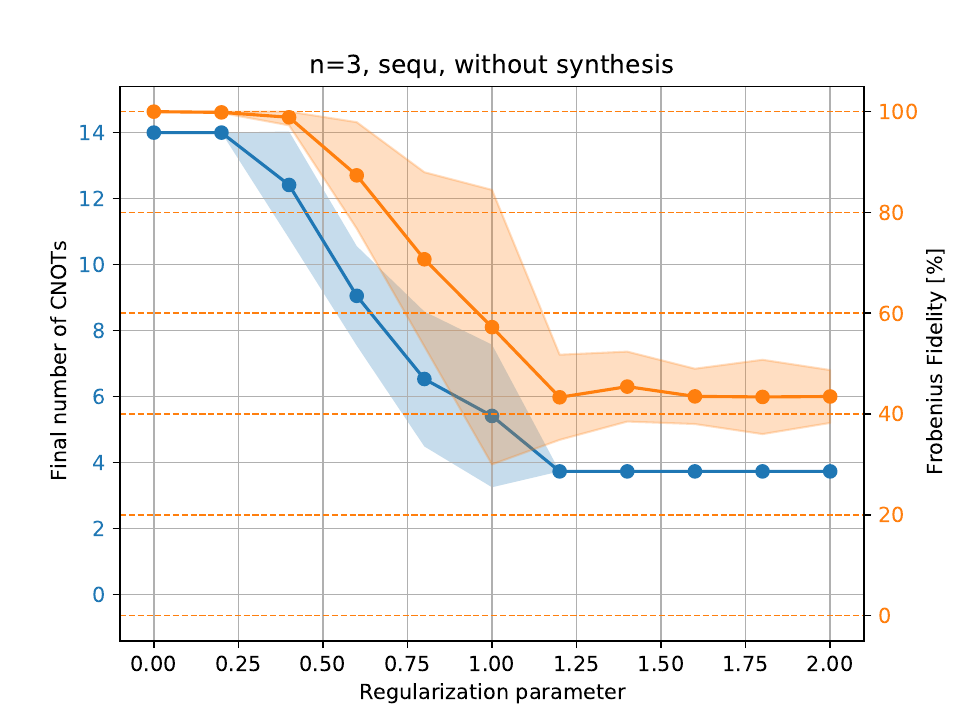}
    \includegraphics[width=.475\linewidth]{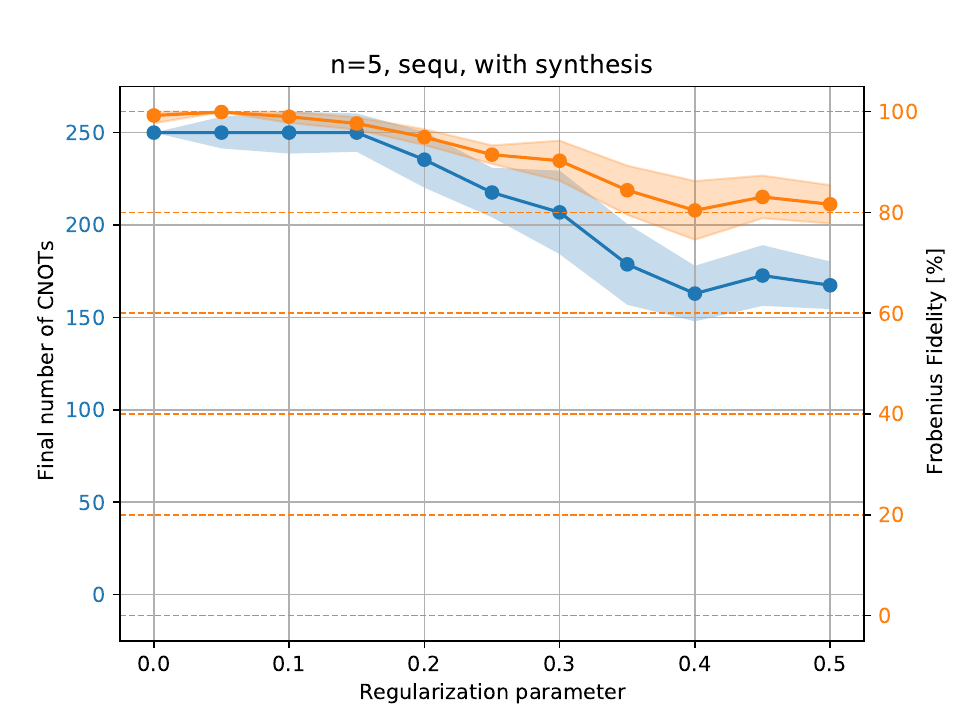}
    \includegraphics[width=.475\linewidth]{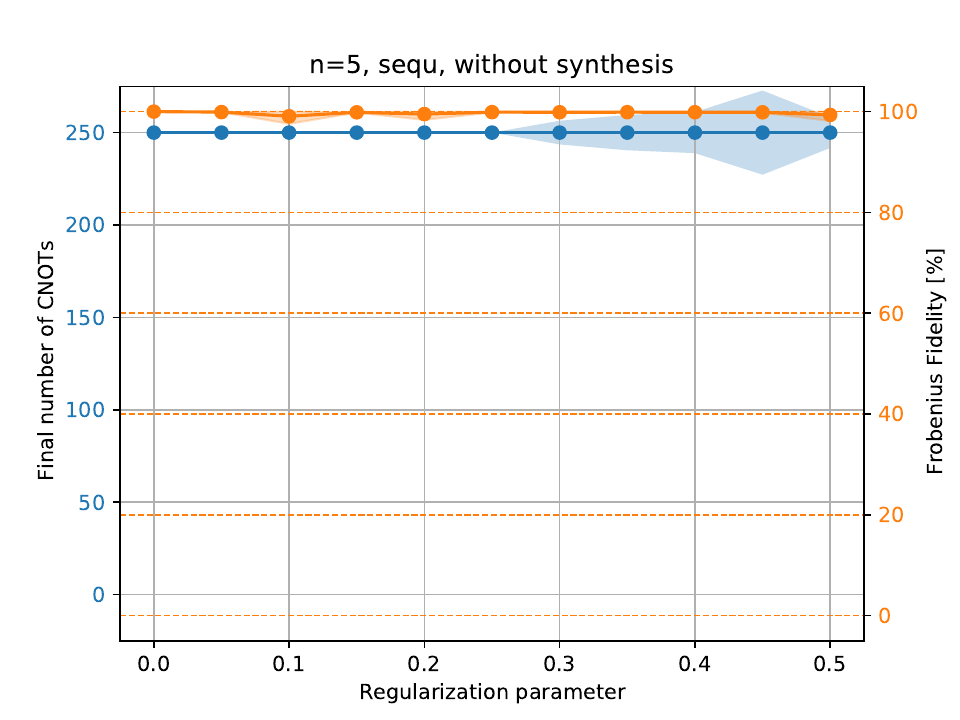}
    \includegraphics[width=.475\linewidth]{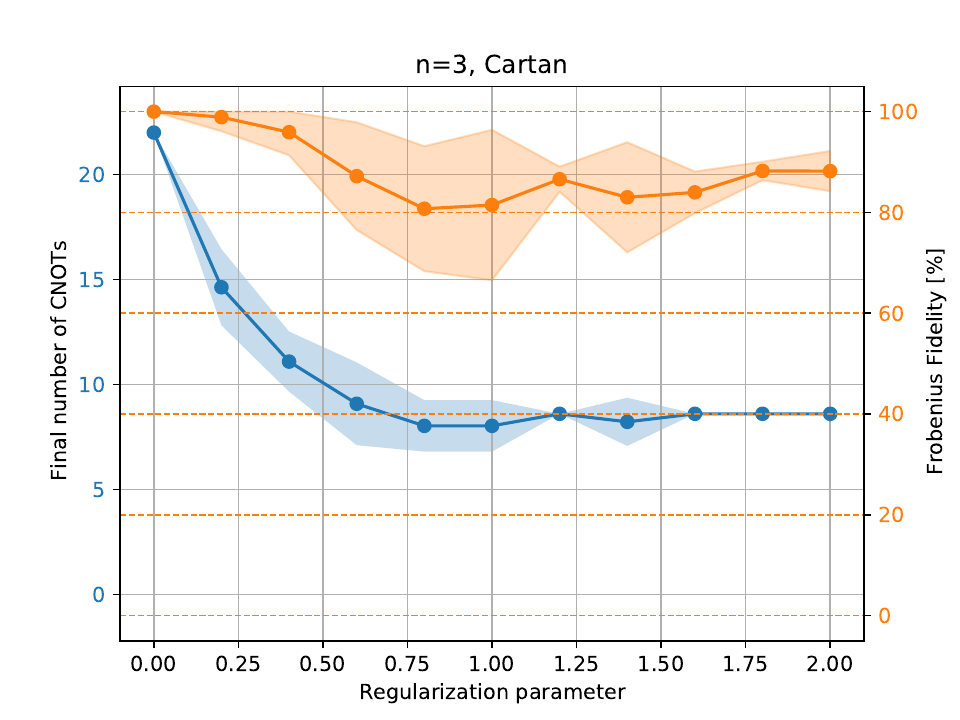}
    \includegraphics[width=.475\linewidth]{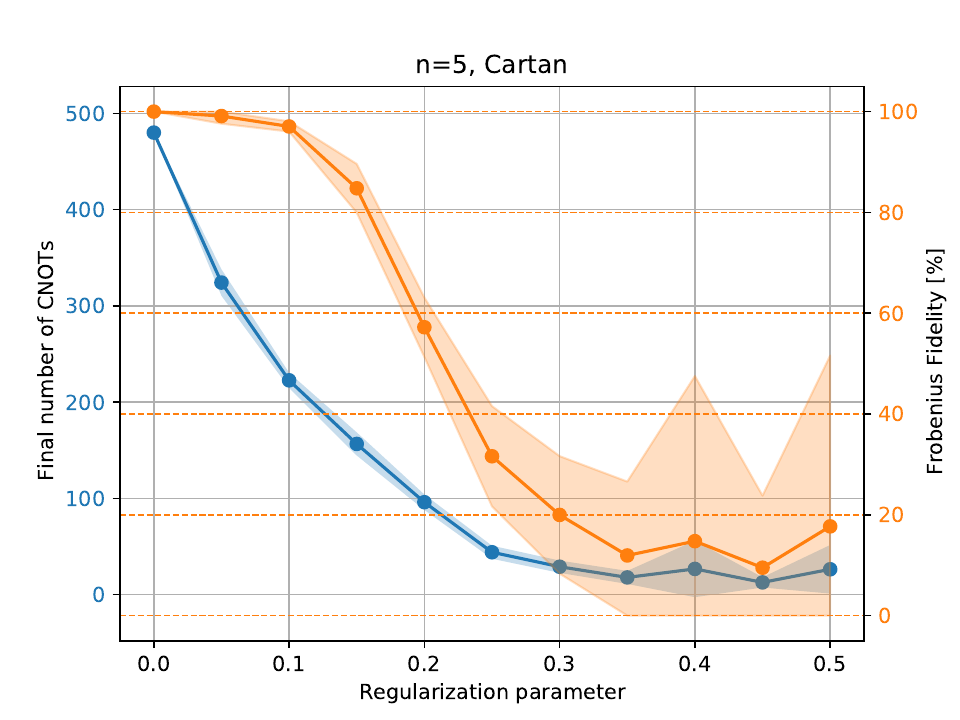}
    \caption{Final number of CNOT units and Frobenius fidelity, for $n=3$ and $n=5$ qubit circuits and different structures. We divide the cases for $\sequ$ for when ``synthesis'' is run, and when it is not. For $\cart$, since hardware constraints are not imposed, ``synthesis'' is always run. \cmt{For all graphs, the $x$-axis represents the regularization parameter $\lambda$, while the $y$-axis represents the compression obtained, in blue, and the Frobenius fidelity, in orange. The continuous line is the average, while the shaded area is one standard deviation.}}
    \label{fig.GL1}
\end{figure}

\subsection{Numerical tests: random unitary matrices, towards [G3]}

We test our compression algorithm on random unitary matrices with $\sequ$ and $\cart$ structures for both $n=3$ and $n=5$ qubit circuits (\cmt{we do not consider $\spin$ because it corresponds to the line connectivity which is not preserved under the synthesis algorithm}). In particular, we randomly initialize the iterates and consider $100$ different random unitary matrices, and we plot both mean and standard deviation of the result. 

\cmt{We report our results in Figure~\ref{fig.GL1}. We start with circuits of $L=14$ and $L=250$ for $\sequ$ and from $L=22$ and $L=528$ for $\cart$, respectively for $n=3$ and $n=5$, and we compress them with different regularization parameters. In the $x$-axis, we can see the regularization parameter $\lambda$ used, while in the $y$-axis, we can note the final number of CNOTs (in blue), as well as the Frobenius fidelity $\bar{F}_{\textrm{F}}(U,V)$ (in orange). The lines correspond to the average, and the shaded areas to one standard deviation. }

As one can appreciate, a small compression is possible for $\sequ$, especially if the ``synthesis'' algorithm is used. \cmt{Note that the three qubit reductions have little affect for $n=5$ and this is true for all $n\ge 5$.}

A better compression, with very high Frobenius fidelity $\bar{F}_{\textrm{F}}(U,V)$, is possible instead for $\cart$, showing that this structure can be tuned to be compressed to TLB ($L=14$ and $L=252$, respectively) without losing fidelity. \cmt{One can see this by looking at the second data point on the left for both $n=3$ and $n=5$, where we obtain a reduction to $L=14$ and $L < 252$ (blue curve), with no practical loss of fidelity (orange curve).} This is encouraging (meaning that the recursive Cartan decomposition can be easily compressed in practice) and supports properties (3) and (4) for $\cart$.

\subsection{Numerical tests: compressing compiled circuits, towards [G3] in Qiskit}

Finally, we move to analyze the effect of the compression algorithm to already compiled circuits in Qiskit, or other compilers. The idea here is to see how one can use the approximate quantum compiler as an add-on \cmt{to} the usual workflow, by allowing the user to trade-off accuracy and \cmt{circuit depth, as defined in Remark~\ref{rem:depth}.} 

We consider the quantum circuits of $3, 4, 5$ qubits in the \cite{database} database. We compile them in Qiskit (opt level of 3 and ``sabre'' as routing and layout) and transform them into \cmt{the parametric circuit of the present paper (see Figure~\ref{fig.expl} for an example)}. \cmt{Then we} use this compiled circuit as a warm start for our compression algorithm, Algorithm~\ref{algo:LASSO} for different $\lambda$'s. This yields the graphs in Figure~\ref{fig:ZUL}, where we can appreciate how $\lambda$ affects compression, as well as Frobenius fidelity. 

\begin{figure}
    \centering
    \includegraphics[width=1.\linewidth]{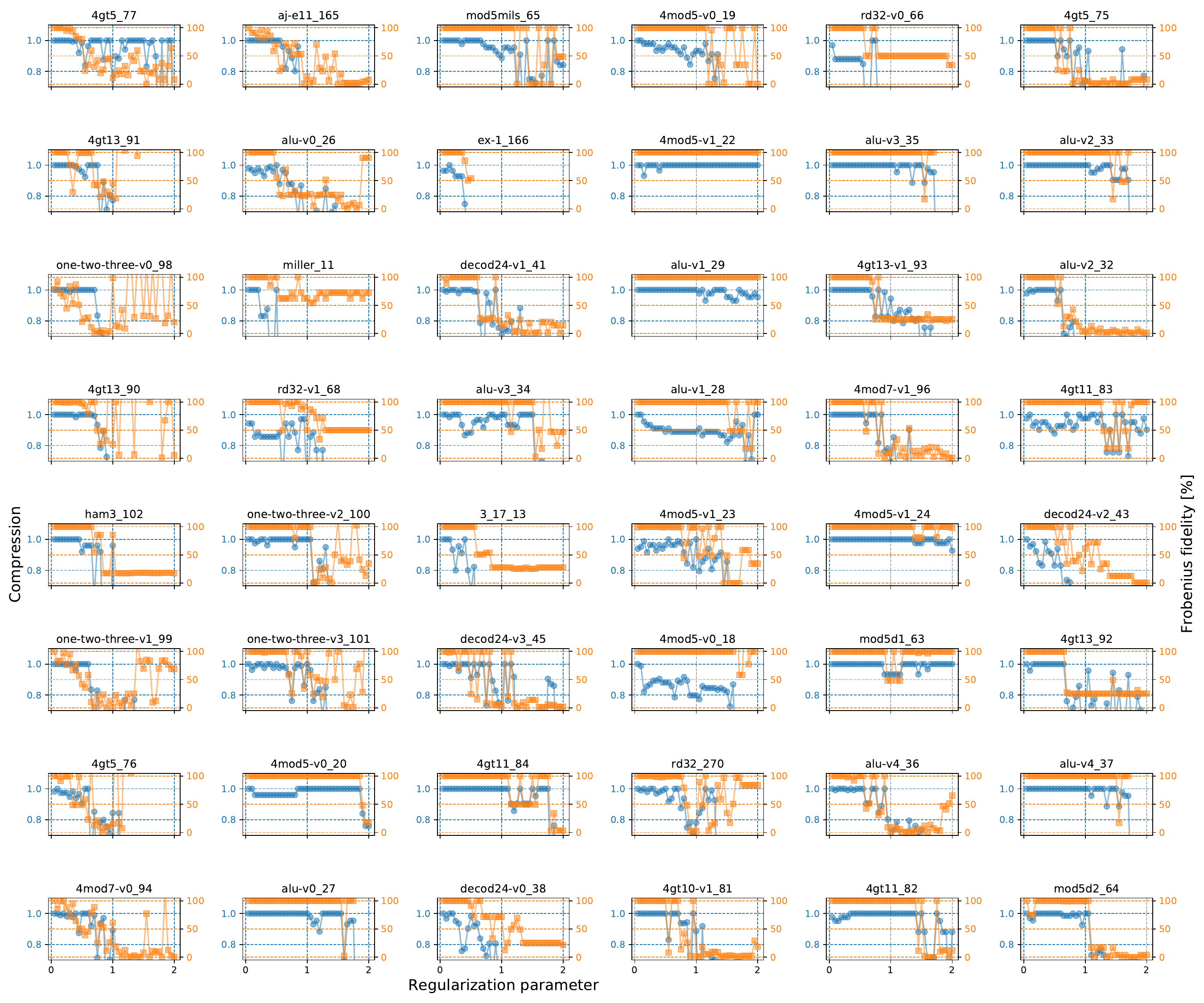}
    \caption{Compression of compiled circuits from the \cite{database} database, for various regularization parameter values. \cmt{For all graphs, the $x$-axis represents the regularization parameter $[0-2]$, while the $y$-axis the compression obtained in blue (normalized to the starting CNOT count), and the Frobenius fidelity in orange. }}
    \label{fig:ZUL}
\end{figure}

A better overview is offered by Table~\ref{tab:my_zul}, where we look at the best compression we can obtain for a Frobenius fidelity as high as $90\%$ and $80\%$, and its associated computational overhead. To compile the table, we have looked at different $\lambda$'s parameters, and we have considered the best compression with fidelity $\ge 90\%, 80\%$. The column `Depth' refers to the depth of the original circuit in the database (which can have gates not in the gate set), the Qiskit depth is the circuit depth once compiled in Qiskit and transformed \cmt{into the parametric circuit}, the compressed depth is the minimal depth that we were able to obtain with Algorithm~\ref{algo:LASSO}, with a fidelity $\ge 90\%$ (or $\ge 80\%$ -- if it is the same depth, we do not report it twice) (by varying $\lambda$), the fidelity column is the Frobenius fidelity of the maximal compressed circuit, while the overhead is the time that the computation of the latter circuit has taken. 

As one can see, \cmt{we can achieve some compression with a small loss in fidelity, and in some cases, without any loss. For some circuits, compression is quite high ($\sim 58\%$ of the Qiskit depth), for others, less so.} 

\cmt{We remark that, ultimately, the compression capabilities are a by-product of Qiskit (or others) compilation properties. If the compiler used can achieve optimal compilation, no further compression can be obtained, no matter how sophisticated the devised algorithm is. While compilation is getting better and better (new exact or heuristic algorithms are proposed at a rapid pace, see for instance~\cite{Giacomo2021}), we believe that our tool can \emph{also} be used to globally gauge if circuits can be further compressed and which part of the circuit is not optimally compiled. This has immediate practical applications in devising better compilation algorithms. For example, further studies are needed to understand what makes `miller\_11' or `decod24-v2\_43' difficult for Qiskit, and how to improve its compiler. }

Finally, more research has to be dedicated \cmt{towards} devising better strategies to select the best parameter $\lambda$ for compression, instead of a random search, and analyzing what happens with the use of different compilers rather than Qiskit. 

\begin{table}
    \centering
    \caption{Best achieved compression for the circuit of the \cite{database} database, along with their Frobenius fidelity and computational overhead. Depth is the original depth; Qiskit depth is the depth once compiled in Qiskit and transformed \cmt{into the parametric circuit}; Compr. depth is the best compression obtain\cmt{ed} with Fr. Fidelity $\ge 90\%$ or $\ge 80\%$ (in the $\ge 80\%$ column, we only report depth if different from the one in $\ge 90\%$), \cmt{while its indicated \% represents the ratio of the compressed depth w.r.t. the Qiskit depth}.}
    \label{tab:my_zul}
    \footnotesize
    \begin{tabular}{c|c|c|c|c|c|c|c}
    \toprule
    &&Qiskit &\multicolumn{2}{|c|}{Fr. Fidelity $\geq 90\%$} & \multicolumn{2}{c|}{Fr. Fidelity $\geq 80\%$} &\\
    File name & Depth & depth & Compr. depth \cmt{(\& \%)} & $\bar{F}_{\mathrm{F}}$ [\%] & Compr. depth \cmt{(\& \%)}  & $\bar{F}_{\mathrm{F}}$ [\%] & Overhead [s] \\ \toprule
    4gt5\_77 & 74 & 137 & 136 \cmt{\,(99\%)}& 93.41 &  126 \cmt{\,(92\%)}& 80.6   & 24.34 \\
4gt13\_91 & 61 & 119 & 119 \cmt{\,(100\%)}& 100.0 &  &  & 20.76 \\
one-two-three-v0\_98 & 82 & 163 & 163 \cmt{\,(100\%)} & 100.0 &  &  & 23.7 \\
4gt13\_90 & 65 & 121 & 120 \cmt{\,(99\%)}& 99.99 &  &  & 20.63 \\
ham3\_102 & 13 & 25 & 24 \cmt{\,(96\%)} & 99.99 &  &  & 1.28 \\
one-two-three-v1\_99 & 76 & 152 & 152 \cmt{\,(100\%)} & 100.0 &  &  & 23.65 \\
4gt5\_76 & 56 & 115 & 111 \cmt{\,(97\%)} & 98.75 &  &  & 20.57 \\
4mod7-v0\_94 & 92 & 175 & 171 \cmt{\,(98\%)} & 95.22 &  &  & 27.19 \\
aj-e11\_165 & 86 & 177 & 177 \cmt{\,(100\%)} & 100.0 &  &  & 36.16 \\
alu-v0\_26 & 49 & 99 & 95 \cmt{\,(96\%)} & 99.99 &  &  & 17.26 \\
miller\_11 & 29 & 65 & 38 \cmt{\,(58\%)} & 99.99 &  &  & 3.1 \\
rd32-v1\_68 & 21 & 35 & 34 \cmt{\,(97\%)}& 92.02 &  31 \cmt{\,(89\%)}& 82.43  & 4.58 \\
one-two-three-v2\_100 & 40 & 79 & 79 \cmt{\,(100\%)} & 100.0 &  &  & 13.82 \\
one-two-three-v3\_101 & 40 & 80 & 79 \cmt{\,(99\%)} & 91.84 &  &  & 16.85 \\
4mod5-v0\_20 & 12 & 25 & 25 \cmt{\,(100\%)}& 100.0 &  &  & 6.01 \\
alu-v0\_27 & 21 & 43 & 41 \cmt{\,(95\%)} & 99.99 &  &  & 7.8 \\
mod5mils\_65 & 21 & 44 & 44 \cmt{\,(100\%)}& 100.0 &  &  & 8.61 \\
ex-1\_166 & 12 & 28 & 26 \cmt{\,(93\%)} & 99.99 &  21 \cmt{\,(75\%)} & 85.34  & 1.28 \\
decod24-v1\_41 & 50 & 94 & 93 \cmt{\,(99\%)} & 99.99 &  &  & 16.32 \\
alu-v3\_34 & 30 & 60 & 60 \cmt{\,(100\%)}& 100.0 &  &  & 11.45 \\
3\_17\_13 & 22 & 45 & 37 \cmt{\,(82\%)} & 99.98 &  &  & 2.11 \\
decod24-v3\_45 & 84 & 157 & 155 \cmt{\,(98\%)}& 98.4 &  &  & 25.86 \\
4gt11\_84 & 11 & 21 & 21  \cmt{\,(100\%)} & 100.0 &  &  & 2.82 \\
decod24-v0\_38 & 30 & 62 & 52 \cmt{\,(84\%)} & 99.99 &  &  & 5.48 \\
4mod5-v0\_19 & 21 & 45 & 41 \cmt{\,(91\%)} & 95.94 &  &  & 8.76 \\
4mod5-v1\_22 & 12 & 29 & 29  \cmt{\,(100\%)} & 100.0 &  &  & 6.06 \\
alu-v1\_29 & 22 & 43 & 41 \cmt{\,(95\%)} & 99.99 &  &  & 6.98 \\
alu-v1\_28 & 22 & 45 & 39 \cmt{\,(87\%)} & 99.99 &  &  & 7.18 \\
4mod5-v1\_23 & 41 & 83 & 76 \cmt{\,(92\%)}& 99.99 &  &  & 14.79 \\
4mod5-v0\_18 & 40 & 84 & 73 \cmt{\,(87\%)} & 99.99 &  &  & 11.06 \\
rd32\_270 & 47 & 96 & 84 \cmt{\,(88\%)} & 99.81 &  81 \cmt{\,(84\%)} & 89.44  & 20.94 \\
4gt10-v1\_81 & 84 & 159 & 159  \cmt{\,(100\%)} & 100.0 &  &  & 29.35 \\
rd32-v0\_66 & 20 & 33 & 33  \cmt{\,(100\%)} & 100.0 &  &  & 3.73 \\
alu-v3\_35 & 22 & 44 & 42 \cmt{\,(95\%)} & 99.99 &  &  & 7.77 \\
4gt13-v1\_93 & 39 & 70 & 70  \cmt{\,(100\%)} & 100.0 &  &  & 14.66 \\
4mod7-v1\_96 & 94 & 164 & 164  \cmt{\,(100\%)} & 100.0 &  &  & 24.4 \\
4mod5-v1\_24 & 21 & 41 & 38 \cmt{\,(93\%)} & 97.4 &  &  & 8.68 \\
mod5d1\_63 & 13 & 30 & 30  \cmt{\,(100\%)} & 100.0 &  &  & 8.92 \\
alu-v4\_36 & 66 & 117 & 117  \cmt{\,(100\%)} & 100.0 &  &  & 20.24 \\
4gt11\_82 & 20 & 42 & 42  \cmt{\,(100\%)} & 100.0 &  &  & 7.18 \\
4gt5\_75 & 47 & 88 & 88  \cmt{\,(100\%)} & 100.0 &  &  & 16.24 \\
alu-v2\_33 & 22 & 42 & 38 \cmt{\,(90\%)} & 99.99 &  &  & 6.96 \\
alu-v2\_32 & 92 & 174 & 174  \cmt{\,(100\%)} & 100.0 &  &  & 25.21 \\
4gt11\_83 & 16 & 41 & 37 \cmt{\,(90\%)} & 99.95 &  &  & 9.18 \\
decod24-v2\_43 & 30 & 65 & 47 \cmt{\,(72\%)} & 99.99 &  &  & 5.23 \\
4gt13\_92 & 38 & 71 & 71  \cmt{\,(100\%)} & 100.0 &  &  & 14.1 \\
alu-v4\_37 & 22 & 44 & 42 \cmt{\,(95\%)} & 99.99 &  &  & 7.73 \\
mod5d2\_64 & 32 & 67 & 67  \cmt{\,(100\%)} & 100.0 &  &  & 14.08 \\
\bottomrule
    \end{tabular}
\end{table}

\section{Conclusions and open points}~\label{sec:concl}

In this paper, we have examined, in deep mathematical and numerical detail, variants of the best approximate quantum compiling problem. We have shown how to \cmt{build} hardware-aware structures and how to optimize over them. While we have presented encouraging results to support various theoretical and numerical properties, several open points are left for future research. In particular, on top of our priority list are (1) to investigate Conjecture \ref{conj:surj} and analytically determine the relationship between approximation error and number of CNOTs (which has been done for 2-qubits using the Weyl chamber \cite{peterson2020two,cross2019validating,crooks2020gates,jurcevic2020demonstration}); (2) to investigate Conjecture \ref{conj:stat}; and (3) to investigate properties (5) and (6) for $\cart$.

\section*{Acknowledgement}
This project has received funding from the Disruptive Technologies Innovation Fund (DTIF), by Enterprise Ireland, under project number DTIF2019-090.

\bibliographystyle{IEEEtran}
\bibliography{Liambib}

\end{document}